\documentclass[12pt, draftclsnofoot, onecolumn]{IEEEtran}
\usepackage{amsmath}
\usepackage{amssymb}
\usepackage{amsfonts}
\usepackage{graphicx}
\usepackage{epsfig}
\usepackage{subfigure}
\usepackage{psfrag}
\usepackage{xcolor}

\title{Joint Power Control and Fronthaul Rate Allocation for Throughput Maximization in OFDMA-based Cloud Radio Access Network
\footnote {L. Liu and S. Bi are with the Department of Electrical
and Computer Engineering, National University of Singapore
(e-mail:\{liu\_liang,bsz@nus.edu.sg\}).}\footnote{R. Zhang is with the
Department of Electrical and Computer Engineering, National
University of Singapore (e-mail:elezhang@nus.edu.sg). He is also
with the Institute for Infocomm Research, A*STAR, Singapore.}}

\author{Liang Liu, Suzhi Bi, and Rui Zhang}

\begin{document}

\maketitle \thispagestyle{empty} \vspace{-0.3in}

\begin{abstract}
The performance of cloud radio access network (C-RAN) is constrained by the limited fronthaul link capacity under future heavy data traffic. To tackle this problem, extensive efforts have been devoted to design efficient signal quantization/compression techniques in the fronthaul to maximize the network throughput. However, most of the previous results are based on information-theoretical quantization methods, which are hard to implement practically due to the high complexity. In this paper, we propose using practical uniform scalar quantization in the uplink communication of an orthogonal frequency division multiple access (OFDMA) based C-RAN system, where the mobile users are assigned with orthogonal sub-carriers for transmission. In particular, we study the joint wireless power control and fronthaul quantization design over the sub-carriers to maximize the system throughput. Efficient algorithms are proposed to solve the joint optimization problem when either information-theoretical or practical fronthaul quantization method is applied. We show that the fronthaul capacity constraints have significant impact to the optimal wireless power control policy. As a result, the joint optimization shows significant performance gain compared with optimizing only wireless power control or fronthaul quantization. Besides, we also show that the proposed simple uniform quantization scheme performs very close to the throughput performance upper bound, and in fact overlaps with the upper bound when the fronthaul capacity is sufficiently large. Overall, our results reveal practically achievable throughput performance of C-RAN for its efficient deployment in the next-generation wireless communication systems.
\end{abstract}


\setlength{\baselineskip}{1.3\baselineskip}
\newtheorem{definition}{\underline{Definition}}[section]
\newtheorem{fact}{Fact}
\newtheorem{assumption}{Assumption}
\newtheorem{theorem}{\underline{Theorem}}[section]
\newtheorem{lemma}{\underline{Lemma}}[section]
\newtheorem{corollary}{\underline{Corollary}}[section]
\newtheorem{proposition}{\underline{Proposition}}[section]
\newtheorem{example}{\underline{Example}}[section]
\newtheorem{remark}{\underline{Remark}}[section]
\newtheorem{algorithm}{\underline{Algorithm}}[section]
\newcommand{\mv}[1]{\mbox{\boldmath{$ #1 $}}}

\section{Introduction}

\subsection{Motivation}
The dramatic increase of mobile data traffic in the recent years has posed imminent challenges to the current cellular systems, requiring higher throughput, larger coverage, and smaller communication delay. The 5G cellular system on the roadmap is expected to achieve up to 1000 times of throughput improvement over today's 4G standard. As a promising candidate for the future 5G standard, cloud radio access network (C-RAN) enables a centralized processing architecture, using multiple relay-like base stations (BSs), named remote radio heads (RRHs), to serve mobile users cooperatively under the coordination of a central unit (CU) \cite{ChinaMobile}. For the practical deployment of C-RAN, a cluster-based C-RAN system is shown in Fig. \ref{fig1}, where the same frequency bands could be reused over non-adjacent or even adjacent C-RAN clusters to increase spectral efficiency through coordination among CUs by applying certain interference management techniques such as dynamic resource allocation \cite{Rui10}. Within each C-RAN cluster, the RRHs are connected to a CU that is further connected to the core network via high-speed fiber fronthaul and backhaul links, respectively. In a C-RAN, a mobile user could be associated with multiple RRHs. However, unlike the BSs in conventional cellular systems which encode/decode user messages locally, the RRHs merely forward the signals to/from the mobile users, while leaving the joint encoding/decoding complexity to a baseband unit (BBU) in the CU. The use of inexpensive and densely deployed RRHs, along with the advanced joint processing mechanism, could significantly improve upon the current 4G system with enhanced scalability, increased throughput and extended coverage.

\begin{figure}
\begin{center}
\scalebox{0.4}{\includegraphics*{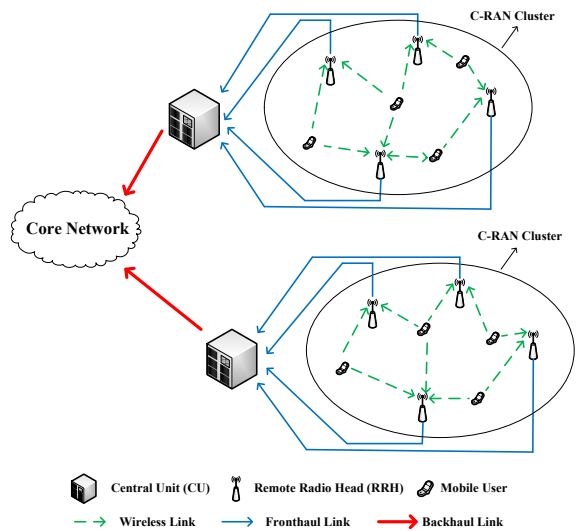}}
\end{center}\vspace{-20pt}
\caption{An illustration of the cluster-based C-RAN in the uplink.}\label{fig1} \vspace{-20pt}
\end{figure}

The distributed antenna system formed by the RRHs enables spectrum efficient spatial division multiple access (SDMA) in C-RAN, which has gained extensive research attentions \cite{Gesbert10}--\cite{Shamai13}. In the uplink communication of an SDMA based C-RAN, all mobile users in the same cluster transmit on the same spectrum and at the same time, while the BBU performs multi-user detection (MUD) to separate the user messages. In practice, however, the implementation of MUD is hurdled by the high computational complexity and the difficulty in signal synchronization as well as perfect channel estimation. Similarly, the downlink communication using SDMA is also of high complexity in the encoding design to mitigate the co-channel interference. With this regard, orthogonal frequency division multiple access (OFDMA) is an alternative candidate for C-RAN because of its efficient spectral usage and yet low encoding/decoding complexity. In OFDMA-based C-RAN systems, users are allocated with orthogonal subcarriers (SCs) free of co-channel interference. In this case, simple maximal-ratio combining (MRC) technique could be performed at the CU over the signals received from different RRHs to decode a user's message transmitted on its designated SC. Moreover, OFDMA is compatible with the current wireless communication systems such as 4G LTE. Considering its potential implementations in future wireless systems and compatibility with the current 4G standards, we consider OFDMA for the cluster-based C-RAN (see Fig. \ref{fig1}) in this paper.

\subsection{Prior Work}

The performance of a C-RAN system is constrained by the fronthaul link capacity. With densely deployed RRHs, the fronthaul traffic generated from a single user signal of MHz bandwidth could be easily scaled up to multiple Gbps \cite{ChinaMobile}. In practice, a commercial fiber link with tens of Gbps capacity could thus be easily overwhelmed even under moderate mobile traffic. To tackle this problem, many signal compression/quantization methods have been proposed to optimize the system performance under fronthaul capacity constraints. Specifically, the so-called ``quantize-and-forward'' scheme is widely adopted for the uplink communication in C-RAN to reduce the communication rates between the BBU and RRHs \cite{Steinberg09}--\cite{Yu13}, where each RRH samples, quantizes and forwards its received signals to the BBU over its fronthaul link. The quantize-and-forward scheme is initially studied in relay channel as an efficient way for the relay to deliver the received signal from the source to the destination \cite{Tse11,Chung11,Kramer08}. In the uplink communication of C-RAN, which can be viewed as a special case of relay channel model with a wireless first-hop link and wired (fiber) second-hop link, quantize-and-forward scheme is studied under an information-theoretical Gaussian test channel model with the uncompressed signals as the input and compressed signals as the output corrupted by an additive Gaussian compression noise. Then, the quantization methods are designed through setting the quantization noise levels at different RRHs to maximize the end-to-end throughput subject to the capacity constraints of individual fronthaul links. Specifically, the optimal quantization design needs to consider the signal correlation across the multiple RRHs, where methods based on distributed source coding, e.g., Wyner-Ziv coding, are widely used to jointly optimize the noise levels at the RRHs (see e.g., \cite{Simeone13}--\cite{Yu13}). Besides, quantization method based on distributed source coding is also studied in the downlink communication of C-RAN in \cite{Shamai13}.

Despite of their respective contributions to the understanding of the theoretical limits of C-RAN, most of the proposed quantization methods are based on information-theoretical models, e.g., Gaussian test channel and distributed source coding, which are practically hard to implement. On one hand, although the quantization noise levels across different RRHs that maximize the end-to-end throughput are found in \cite{Steinberg09}--\cite{Yu13} under different system setups, it is still unknown how to practically design quantization codebook at each RRH to achieve the required quantization noise level for the Gaussian test channel model. On the other hand, the decompression complexity of distributed source coding grows exponentially with the number of sources (e.g., RRHs in the uplink communication). In practice, the complexity can be prohibitively high in a C-RAN with a large number of cooperating RRHs. Therefore, it still remains as a question about the practically achievable throughput of C-RAN using practical quantization methods, such as uniform scalar or vector quantization used in common A/D modules \cite{Gray98}, which are independently applied over RRHs.

Furthermore, most of the existing works (e.g., \cite{Steinberg09}--\cite{Yu13}) only study signal compression methods in C-RAN under fixed wireless resource allocation. However, the end-to-end performance of C-RAN is determined by both the wireless and fronthaul links. In an OFDMA system, transmit power allocation over frequency SCs directly determines the spectral efficiency of wireless link. For an OFDMA-based system without fronthaul constraint, the optimal power allocation problem is extensively studied, e.g., it follows the celebrated water-filling policy for a single user case \cite{Goldsmith05}. However, the behavior of optimal SC power allocation in a fronthaul constrained system like C-RAN is still unknown to the authors' best knowledge.

\subsection{Main Contribution}
In this paper, we address the above problems in an OFDMA-based C-RAN. In particular, we consider using simple uniform scalar quantization instead of the information-theoretical quantization method based on Gaussian test channel, and propose joint wireless power control and fronthaul rate allocation design to maximize the system throughput performance. Our main contributions are summarized as follows:

\begin{itemize}
\item In the uplink communication of an OFDMA-based C-RAN, we derive the end-to-end sum-rate of all the users subject to each RRH's fronthaul capacity constraint achieved by a simple uniform scalar quantization at each RRH together with independent compression among RRHs. Different from prior works based on Gaussian test channel model, this provides for the first time an achievable rate result for C-RAN with a practically implementable quantization method.

\item With the derived rate under uniform scalar quantization, we formulate the optimization problem of joint wireless power control and fronthaul rate allocation to maximize the sum-rate performance in OFDMA based C-RAN. We also formulate the problem based on the Gaussian test channel model to obtain performance benchmark. Efficient algorithms are proposed to solve the formulated joint optimization problems based on the alternating optimization technique.

\item By investigating the single-user and single-RRH special case, we obtain important insights on the optimal wireless power control and fronthaul rate allocation over SCs. For example, with a fixed fronthaul rate allocation, we show that the optimal power allocation over SCs is a threshold based policy depending on the channel power of a SC, i.e., no power is allocated to a SC if the channel power is below the threshold. Interestingly, we find that the power allocation under fronthaul rate constraint in general does not follow a water-filling policy that always allocates more power to SC with higher channel power. The inconsistency is especially evident in low-fronthaul-rate region, where the SC with the highest channel power may receive the least transmit power, and vice versa. We also theoretically quantify the performance gap between the proposed simple uniform quantization scheme from the throughput upper (cut-set) bound. By simulations we show that the throughput performance of the simple uniform quantization scheme is very close to the performance upper bound, and in fact overlaps with the upper bound when the fronthaul capacity is sufficiently large.
\end{itemize}

\subsection{Organization}
The rest of this paper is organized as follows. We first introduce in Sections \ref{sec:System Model} and \ref{sec:Two Scalar Quantization Models} the system model of C-RAN and the quantization techniques used in the fronthaul signal processing, respectively. In Section \ref{sec:Problem Formulation}, we formulate the end-to-end sum-rate maximization problems for both the Gaussian test channel and uniform scalar quantization models. Sections \ref{sec:Special Case: Single User and Single RRH} and \ref{sec:General Case: Multiple Users and Multiple RRHs} solve the formulated problems for the special case of single-user and single-RRH and general case of multi-user and multi-RRH, respectively. Finally, we conclude the paper and point out some directions for future work in Section \ref{sec:Conclusion}.

\section{System Model}\label{sec:System Model}

We consider the uplink of a clustered C-RAN. As shown in Fig. \ref{fig1}, each cluster consists of one BBU, $M$ single-antenna RRHs, denoted by the set $\mathcal{M}=\{1,\cdots,M\}$, and $K$ single-antenna users, denoted by the set $\mathcal{K}=\{1,\cdots,K\}$. It is assumed that each RRH $m$, $\forall m\in \mathcal{M}$, is connected to the BBU through a noiseless wired fronthaul link of capacity $\bar{T}_m$ bps. In the uplink, each RRH receives user signals over the wireless link and forwards to the BBU via its fronthaul link. Then, the BBU jointly decodes the users' messages based on the signals from all the RRHs within the cluster and forwards the decoded information to the core network through a backhaul link. The detailed signal models in the wireless and the fronthaul links are introduced in the following.

\subsection{OFDMA-based Wireless Transmission}\label{sec:OFDMA-based Franthaul Transmission}

In this paper, we consider OFDMA-based uplink information transmission between the $K$ users and the $M$ RRHs over a wireless link of a $B$Hz total bandwidth equally divided into $N$ SCs. The SC set is denoted by $\mathcal{N}=\{1,\cdots,N\}$. It is assumed that each SC $n\in \mathcal{N}$ is only allocated to one user. Denote $\Omega_k$ as the set of SCs allocated to user $k$, $\forall k\in\mathcal{K}$. In practice, dynamic SC allocation could be used to enhance the spectral efficiency by assigning SCs to users of favorable wireless link conditions, e.g., allocating a SC to the user with the highest signal-to-interference-plus-noise ratio (SINR).  However, as an initial attempt to understand the joint design of the wireless resource allocation and fronthaul rate allocation in fronthaul constrained C-RAN, it is assumed for simplicity in this paper that the SC allocations among users, i.e., $\Omega_k$'s, are pre-determined. The interesting case with dynamic SC allocation is left for future study.

\begin{figure}
\begin{center}
\scalebox{0.6}{\includegraphics*{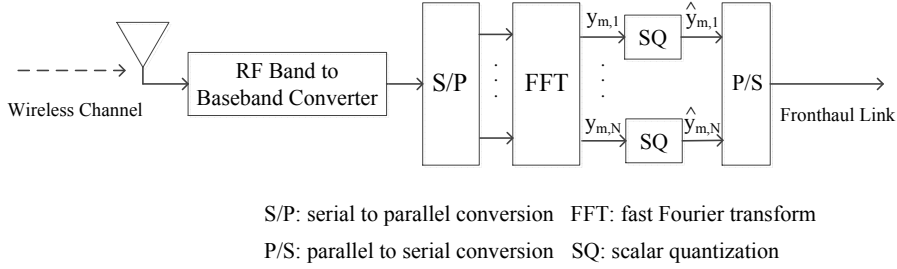}}
\end{center}\vspace{-20pt}
\caption{The structure of signal processing at an RRH.}\label{fig2} \vspace{-20pt}
\end{figure}

Specifically, in the uplink each user $k$, $\forall k\in \mathcal{K}$, first generates an OFDMA modulated signal over its assigned SCs and then transmits to the RRHs in the same cluster. As shown in Fig. \ref{fig2}, each RRH $m$, $\forall m\in \mathcal{M}$, first downconverts the received RF signals to the baseband, then transforms the serial baseband signals to the parallel ones, and demodulates the parallel signals into $N$ streams by performing fast Fourier transform (FFT). Suppose that $n\in \Omega_k$, then the equivalent baseband complex symbol received by RRH $m$ at SC $n$ can be expressed as
\begin{align}\label{eqn:received signal}
y_{m,n}=h_{m,k,n}\sqrt{p_{k,n}}s_{k,n}+z_{m,n},
\end{align}where $s_{k,n}\sim\mathcal{CN}(0,1)$ denotes the transmit symbol of user $k$ at SC $n$ (which is modelled as a circularly symmetric complex Gaussian random variable with zero-mean and unit-variance), $p_{k,n}$ denotes the transmit power of user $k$ at SC $n$, $h_{m,k,n}$ denotes the channel from user $k$ to RRH $m$ at SC $n$, and $z_{m,n}\sim\mathcal{CN}(0,\sigma_{m,n}^2)$ denotes the aggregation of additive white Gaussian noise (AWGN) and (possible) out-of-cluster interference at RRH $m$ at SC $n$. It is assumed that $z_{m,n}$'s are independent over $m$ and $n$.

\subsection{Quantize-and-Forward Processing at RRH}\label{sec:Quantize-and-Forwards-based Backhaul Transmission}

To forward the baseband symbols $y_{m,n}$'s to the BBU via the fronthaul links, the so-called ``quantize-and-forward'' scheme is applied, where each RRH first quantizes its baseband received signal and then sends the corresponding digital codewords to the BBU. Specifically, since at each RRH the received symbols at all the SCs are independent with each other and we assume independent signal quantization at different RRHs, a simple scalar quantization on $y_{m,n}$'s is optimal as shown in Fig. \ref{fig2}. The baseband quantized symbol of $y_{m,n}$ is then given by
\begin{align}\label{eqn:quantized signal}
\tilde{y}_{m,n}=y_{m,n}+e_{m,n}=h_{m,k,n}\sqrt{p_{k,n}}s_{k,n}+z_{m,n}+e_{m,n},
\end{align}where $e_{m,n}$ denotes the quantization error for the received symbol $y_{m,n}$ with zero mean and variance $q_{m,n}$. Note that $e_{m,n}$'s are independent over $n$ due to scalar quantization at each SC, and over $m$ due to independent compression among RRHs. Then, each RRH transforms the parallel encoded bits $\hat{y}_{m,n}$'s into the serial ones and sends them to the BBU via its fronthaul link for joint information decoding.

After collecting the digital codewords, the BBU first recovers the baseband quantized symbols $\tilde{y}_{m,n}$'s based on the quantization codebooks used by each RRH. Then, to decode $s_{k,n}$, the BBU applies a linear combining on the quantized symbols at SC $n$ collected from all RRHs:
\begin{align}\label{eqn:beamforming}
\hat{s}_{k,n}=\mv{w}_n^H\tilde{\mv{y}}_n=\mv{w}_n^H\mv{h}_{k,n}\sqrt{p_{k,n}}s_{k,n}+\mv{w}_n^H\mv{z}_n+\mv{w}_n^H\mv{e}_{n}, ~~~ n\in \Omega_k, ~ k=1,\cdots,K,
\end{align}where $\tilde{\mv{y}}_n=[\tilde{y}_{1,n},\cdots,\tilde{y}_{M,n}]^T$, $\mv{h}_{k,n}=[h_{1,k,n},\cdots,h_{M,k,n}]^T$, $\mv{z}_n=[z_{1,n},\cdots,z_{M,n}]^T$, and $\mv{e}_{n}=[e_{1,n},\cdots,e_{M,n}]^T$. According to (\ref{eqn:beamforming}), the SNR for decoding $s_{k,n}$ is expressed as
\begin{align}\label{eqn:SINR}
\gamma_{k,n}=\frac{p_{k,n}|\mv{w}_n^H\mv{h}_{k,n}|^2}{\mv{w}_n^H\left({\rm diag}(\sigma_{1,n}^2,\cdots,\sigma_{M,n}^2)+{\rm diag}(q_{1,n},\cdots,q_{M,n})\right)\mv{w}_n}, ~~~ n\in \Omega_k, ~ k=1,\cdots,K,
\end{align}where ${\rm diag}(\mv{a})$ denotes a diagonal matrix with the main diagonal given by vector $\mv{a}$. It can be shown that the optimal combining weights that maximize $\gamma_{k,n}$'s are obtained from the well-known MRC \cite{Goldsmith05}:
\begin{align}\label{eqn:optimal beamforming}
\mv{w}_{n}^\ast=\left({\rm diag}(\sigma_{1,n}^2,\cdots,\sigma_{M,n}^2)+{\rm diag}(q_{1,n},\cdots,q_{M,n})\right)^{-1}\mv{h}_{k,n}, ~~~ n=1,\cdots,N.
\end{align}With the above MRC receiver, $\gamma_{k,n}$ given in (\ref{eqn:SINR}) reduces to
\begin{align}\label{eqn:optimal SINR}
\gamma_{k,n}=\sum\limits_{m=1}^M\frac{|h_{m,k,n}|^2p_{k,n}}{\sigma_{m,n}^2+q_{m,n}}, ~~~ n\in \Omega_k, ~ k=1,\cdots,K.
\end{align}

\section{Quantization Schemes}\label{sec:Two Scalar Quantization Models}

The key issue to implement the quantize-and-forward scheme introduced in Section \ref{sec:System Model} is how each RRH should quantize its received signal at each SC in practice. In this section, we first study a theoretical quantization model by viewing (\ref{eqn:quantized signal}) as a test channel and derive its achievable sum-rate based on the rate-distortion theory, which can serve as a performance upper bound. Then, we investigate the practical uniform scalar quantization scheme in details, which can be easily applied at each RRH, and derive the corresponding achievable end-to-end sum-rate.

\subsection{Gaussian Test Channel}\label{sec:Gaussian Test Channel Model}

In this subsection, we assume that the quantization errors given in (\ref{eqn:quantized signal}) are Gaussian distributed, i.e., $e_{m,n}\sim\mathcal{CN}(0,q_{m,n})$, $\forall m, n$. With Gaussian quantization errors, (\ref{eqn:quantized signal}) can be viewed as a Gaussian test channel \cite{Gamal11}. As a result, to forward the received data at SC $n$, the transmission rate in RRH $m$'s fronthaul link is expressed as \cite{Gamal11}
\begin{align}\label{eqn:fronthaul link SC}
T_{m,n}^{(G)}=\frac{B}{N}\log_2\left(1+\frac{|h_{m,k,n}|^2p_{k,n}+\sigma_{m,n}^2}{q_{m,n}}\right).
\end{align}Since quantization is performed at each RRH independently,  $\{y_{m,1},\cdots,y_{m,N}\}$ can be reliably transmitted to the BBU if and only if \cite{Kramer08}
\begin{align}\label{eqn:fronthaul link}
T_m^{(G)}=\sum\limits_{n=1}^NT_{m,n}^{(G)}=\frac{B}{N}\sum\limits_{k=1}^K\sum\limits_{n\in \Omega_k} \log_2\left(1+\frac{|h_{m,k,n}|^2p_{k,n}+\sigma_{m,n}^2}{q_{m,n}}\right)\leq \bar{T}_m, ~~~ m=1,\cdots,M.
\end{align}

Next, consider the end-to-end performance of the users. With Gaussian noise in (\ref{eqn:quantized signal}), the achievable rate of user $k$ at SC $n$ is expressed as
\begin{align}
R_{k,n}^{(G)}&=\frac{B}{N}\log_2(1+\gamma_{k,n})=\frac{B}{N}\log_2\left(1+\sum\limits_{m=1}^M\frac{|h_{m,k,n}|^2p_{k,n}}{\sigma_{m,n}^2+q_{m,n}}\right)\nonumber \\ &\overset{(a)}{=}\frac{B}{N}\log_2\left(1+\sum\limits_{m=1}^M\frac{|h_{m,k,n}|^2p_{k,n}}{\sigma_{m,n}^2+\frac{|h_{m,k,n}|^2p_{k,n}+\sigma_{m,n}}{2^{NT_{m,n}^{(G)}/B}-1}}\right), \label{eqn:test channel user rate}
\end{align}where $(a)$ is obtained by substituting $q_{m,n}$ by $T_{m,n}$ according to (\ref{eqn:fronthaul link SC}). Notice that as the allocated fronthaul rate $T_{m,n}^{(G)} \rightarrow 0$ (versus $\infty$), the achievable end-to-end rate in (\ref{eqn:test channel user rate}) converges to zero (or that of the wireless link capacity). Then, the achievable throughput of all users is expressed as
\begin{align}\label{eqn:test channel rate}
R_{{\rm sum}}^{(G)} =\sum\limits_{k=1}^K\sum\limits_{n\in \Omega_k}R_{k,n}^{(G)}=\frac{B}{N}\sum\limits_{k=1}^K\sum\limits_{n\in \Omega_k} \log_2\left(1+\sum\limits_{m=1}^M\frac{|h_{m,k,n}|^2p_{k,n}}{\sigma_{m,n}^2+\frac{|h_{m,k,n}|^2p_{k,n}+\sigma_{m,n}^2}{2^{NT_{m,n}^{(G)}/B}-1}}\right).
\end{align}From (\ref{eqn:test channel rate}), it is clearly seen that the sum-rate performance depends on both the users' power allocations, $\{p_{k,n}\}$, and the RRHs' fronthaul rate allocations, $\{T_{m,n}^{(G)}\}$, over the SCs.

\subsection{Uniform Scalar Quantization}\label{sec:Uniform Quantization Model}

In practice, it is very difficult to find the quantization codebooks to achieve the throughput given in (\ref{eqn:test channel rate}) subject to the fronthaul capacity constraints given in (\ref{eqn:fronthaul link}). In this subsection, we consider using practical uniform scalar quantization technique at each RRH and derive the achievable sum-rate.

\begin{figure}
\begin{center}
\scalebox{0.5}{\includegraphics*{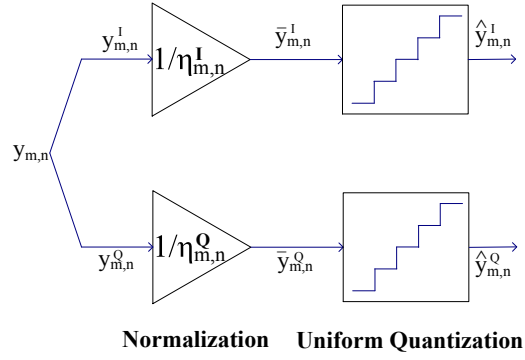}}
\end{center}\vspace{-20pt}
\caption{Schematic of uniform scalar quantization.}\label{fig3}\vspace{-20pt}
\end{figure}

A typical method to implement the uniform quantization is via separate in-phase/quadrature (I/Q) quantization, where the architecture is shown in Fig. \ref{fig3}. Specifically, the received complex symbol $y_{m,n}$ given in (\ref{eqn:received signal}) could be presented by its I and Q parts:
\begin{align}\label{eqn:inphase and Q}
y_{m,n}=y_{m,n}^I+jy_{m,n}^Q,  ~~~ \forall m,n,
\end{align}where $j^2=-1$, and the I-branch symbol $y_{m,n}^I$ and Q-branch symbol $y_{m,n}^Q$ are both real Gaussian random variables with zero mean and variance $(|h_{m,k,n}|^2p_{k,n}+\sigma_{m,n}^2)/2$. As a result, each RRH $m$ first normalizes its I-branch and Q-branch symbols at SC $n$ to $\bar{y}_{m,n}^I$ and $\bar{y}_{m,n}^Q$ by factors $\eta_{m,n}^I$ and $\eta_{m,n}^Q$, and then implements uniform scalar quantization to $\bar{y}_{m,n}^I$ and $\bar{y}_{m,n}^Q$ with $D_{m,n}$ quantization bits, separately. For conciseness, we summarize the implementation details of the uniform scalar quantization in Appendix \ref{appendix1}.

In the following, we present the end-to-end achievable throughput of all users subject to the fronthaul capacity constraints under the uniform scalar quantization technique described in Appendix \ref{appendix1}.

\begin{proposition}\label{fronthaul}
With the uniform scalar quantization scheme, the transmission rate from RRH $m$ to the BBU in its fronthaul link is given as
\begin{align}\label{eqn:fronthaul link uniform quantization}
T_m^{(U)}=\sum\limits_{n=1}^NT_{m,n}^{(U)}\leq \bar{T}_m, ~~~ m=1,\cdots,M,
\end{align}where $T_{m,n}^{(U)}$ denotes the transmission rate in RRH $m$'s fronthaul link to forward its received data at SC $n$, i.e.,
\begin{align}\label{eqn:fronthaul link uniform quantization SC}
T_{m,n}^{(U)}=\frac{2BD_{m,n}}{N}, ~~~ \forall m,n.
\end{align}
\end{proposition}

\begin{proof}
Please refer to Appendix \ref{appendix2}.
\end{proof}

\begin{proposition}\label{rate}
With the uniform scalar quantization scheme, an achievable end-to-end throughput of all users is expressed as
\begin{align}\label{eqn:uniform quantization sum-rate}
R_{{\rm sum}}^{(U)}=\sum\limits_{k=1}^K \sum\limits_{n\in \Omega_k} R_{k,n}^{(U)},\end{align}where the achievable rate of user $k$ at SC $n$ is expressed as
\begin{align}\label{eqn:new end to end rate}
R_{k,n}^{(U)}=\frac{B}{N}\log_2\left(1+\sum\limits_{m=1}^M\frac{|h_{m,k,n}|^2p_{k,n}}{\sigma_{m,n}^2+3(|h_{m,k,n}|^2p_{k,n}+\sigma_{m,n}^2)2^{-\frac{NT_{m,n}^{(U)}}{B}}}\right).
\end{align}
\end{proposition}

\begin{proof}
Please refer to Appendix \ref{appendix3}.
\end{proof}

Notice that (\ref{eqn:new end to end rate}) holds when $T_{m,n}^{(U)} \geq (2B)/N$ (i.e., $D_{m,n} \geq 1$) according to (\ref{eqn:fronthaul link uniform quantization SC}). Similar to (\ref{eqn:test channel rate}) for the ideal case of Gaussian compression, the sum-rate in (\ref{eqn:uniform quantization sum-rate}) with the uniform scalar quantization also jointly depends on both the users' power allocations, $\{p_{k,n}\}$, and the RRHs' fronthaul rate allocations, $\{T_{m,n}^{(U)}\}$, over the SCs. Furthermore, given the same set of power and fronthaul rate allocations, the achievable rate in (\ref{eqn:uniform quantization sum-rate}) is always strictly less than that in (\ref{eqn:test channel rate}) provided that $T_{m,n}^{(G)}=T_{m,n}^{(U)}\geq (2B)/N$, $\forall m, n$.

\section{Problem Formulation}\label{sec:Problem Formulation}

In this paper, given the wireless bandwidth $B$, each user $k$'s SC allocation $\Omega_k$'s as well as transmit power constraint $\bar{P}_k$'s, and each RRH $m$'s fronthaul link capacity $\bar{T}_m$'s, we aim to maximize the end-to-end throughput of all the users subject to each RRH's fronthaul link capacity constraint by jointly optimizing the wireless power control and fronthaul rate allocation. Specifically, for the benchmark scheme, i.e., the theoretical Gaussian test channel based scheme in Section \ref{sec:Gaussian Test Channel Model}, we are interested in solving the following problem.
\begin{align*}\mathrm{(P1)}:~\mathop{\mathtt{Maximize}}_{\{p_{k,n},T_{m,n}^{(G)}\}} & ~~~ R_{{\rm sum}}^{(G)}  \\
\mathtt {Subject \ to} & ~~~ \sum\limits_{n=1}^NT_{m,n}^{(G)}\leq \bar{T}_m, ~~~ \forall m \in \mathcal{M}, \\ & ~~~ \sum\limits_{n\in \Omega_k}p_{k,n}\leq \bar{P}_k, ~~~ \forall k \in \mathcal{K},
\end{align*}where $R_{{\rm sum}}^{(G)}$ is given in (\ref{eqn:test channel rate}) and $T_{m,n}^{(G)}$ is given in (\ref{eqn:fronthaul link SC}). Furthermore, for the proposed uniform scalar quantization based scheme in Section \ref{sec:Uniform Quantization Model}, we are interested in solving the following problem.
\begin{align*}\mathrm{(P2)}:~\mathop{\mathtt{Maximize}}_{\{p_{k,n},T_{m,n}^{(U)}\}} & ~~~ R_{{\rm sum}}^{(U)} \\
\mathtt {Subject \ to} & ~~~ \sum\limits_{n=1}^N T_{m,n}^{(U)}\leq \bar{T}_m, ~~~ \forall m \in \mathcal{M}, \\ & ~~~ \sum\limits_{n\in \Omega_k} p_{k,n}\leq \bar{P}_k, ~~~ \forall k \in \mathcal{K}, \\ & ~~~ T_{m,n}^{(U)} =\frac{2BD_{m,n}^{(U)}}{N}, ~ D_{m,n}\in\{1,2,\cdots\}  ~ {\rm is \ an \ integer}, ~ \forall m \in \mathcal{M}, ~ \forall n \in \mathcal{N},
\end{align*}where $R_{{\rm sum}}^{(U)}$ is given in (\ref{eqn:uniform quantization sum-rate}) and $T_{m,n}^{(U)}$ is given in (\ref{eqn:fronthaul link uniform quantization SC}).

Recall that with the same rate allocations in the fronthaul links for the two schemes, i.e., $T_{m,n}^{(G)}=T_{m,n}^{(U)}\geq (2B)/N$, $\forall m,n$, $R_{{\rm sum}}^{(G)}$ in (\ref{eqn:test channel rate}) is always larger than $R_{{\rm sum}}^{(U)}$ given in (\ref{eqn:uniform quantization sum-rate}). Furthermore, uniform scalar quantization requires that the fronthaul rate allocated at each SC must be an integer multiplication of $(2B)/N$. Due to the above two reasons, in general the optimal value of problem (P2) is smaller than that of problem (P1), i.e., $R_{{\rm sum}}^{(U)}<R_{{\rm sum}}^{(G)}$. It is also worth noting that user association is also determined from solving problems (P1) and (P2), since if with the obtained solution we have $T_{m,n}=0$, $\forall n\in \Omega_k$, RRH $m$ will not quantize and forward user $k$'s signal to the BBU for decoding, or equivalently RRH $m$ does not serve that user at all.

%
%
%

It can be also observed that both problems (P1) and (P2) are non-convex since their objective functions are not concave over $p_{k,n}$'s and $T_{m,n}$'s; thus, it is difficult to obtain their optimal solutions in general. In the following two sections, we first study the special case of problems (P1) and (P2) with one user and one RRH to shed some light on the mutual influence between the wireless power allocation and fronthaul rate allocation, and then propose efficient algorithms to solve problems (P1) and (P2) for the general case of multiple users and multiple RRHs.

\section{Special Case: Single User and Single RRH}\label{sec:Special Case: Single User and Single RRH}

In this section, we study problems (P1) and (P2) for the special case of $K=1$ and $M=1$. For convenience, in the rest of this section we omit the subscripts of $k$ and $m$ in all the notations in problems (P1) and (P2).

\subsection{Gaussian Test Channel}\label{sec:power control and fronthaul rate allocation for the case of one user and one RRH}

It can be shown that problem (P1) is still a non-convex problem for the case of $K=1$ and $M=1$. In this subsection, we propose to apply the alternating optimization technique to solve this problem. Specifically, first we fix the fronthaul rate allocation $T_n^{(G)}=\hat{T}_n^{(G)}$'s in problem (P1) and optimize the wireless power allocation by solving the following problem.
\begin{align}\mathop{\mathtt{Maximize}}_{\{p_n\}} & ~~~ \frac{1}{N}\sum\limits_{n=1}^N\log_2\left(1+\frac{|h_n|^2p_n}{\sigma_n^2+\frac{|h_n|^2p_n+\sigma_n^2}{2^{N\hat{T}_n^{(G)}/B}-1}}\right) \nonumber \\
\mathtt {Subject \ to} & ~~~ \sum\limits_{n=1}^Np_n \leq \bar{P}. \label{eqn:p2}
\end{align}Let $\{\hat{p}_n\}$ denote the optimal solution to problem (\ref{eqn:p2}). Next, we fix the wireless power allocation $p_n=\hat{p}_n$'s in problem (P1) and optimize the fronthaul rate allocation by solving the following problem.
\begin{align}\mathop{\mathtt{Maximize}}_{\{T_n^{(G)}\}} & ~~~ \frac{1}{N}\sum\limits_{n=1}^N\log_2\left(1+\frac{|h_n|^2\hat{p}_n}{\sigma_n^2+\frac{|h_n|^2\hat{p}_n+\sigma_n^2}{2^{NT_n^{(G)}/B}-1}}\right) \nonumber \\
\mathtt {Subject \ to} & ~~~ \sum\limits_{n=1}^N T_n^{(G)} \leq \bar{T}. \label{eqn:p3}
\end{align}Let $\{\hat{T}_n^{(G)}\}$ denote the optimal solution to problem (\ref{eqn:p3}). The above update of $\{p_n\}$ and $\{T_n^{(G)}\}$ is iterated until convergence. In the following, we show how to solve problems (\ref{eqn:p2}) and (\ref{eqn:p3}), respectively.

First, it can be shown that the objective function of problem (\ref{eqn:p2}) is concave over $p_n$'s. As a result, problem (\ref{eqn:p2}) is a convex problem, and thus can be efficiently solved by the Lagrangian duality method \cite{Boyd04}. We then have the following proposition.

\begin{proposition}\label{proposition1}
The optimal solution to problem (\ref{eqn:p2}) is expressed as
\begin{align}\label{eqn:opt1}
\hat{p}_n=\left\{\begin{array}{ll}\frac{-\alpha_n+\sqrt{\alpha_n^2-4\eta_n}}{2}, & {\rm if} ~ \frac{|h_n|^2}{\sigma_n^2}>f_n(\hat{T}_n^{(G)}), \\ 0, & {\rm otherwise}.
\end{array}\right. ~~~ n=1,\cdots,N, \end{align}where \begin{align}
& \alpha_n=\frac{\sigma_n^2(2^{\frac{N\hat{T}_n^{(G)}}{B}}+1)}{|h_n|^2}, \label{eqn:alpha}\\
& \eta_n=\frac{\sigma_n^42^{\frac{N\hat{T}_n^{(G)}}{B}}}{|h_n|^4}-\frac{\sigma_n^2(2^{\frac{N\hat{T}_n^{(G)}}{B}}-1)}{\lambda N|h_n|^2\ln 2}, \label{eqn:eta} \\ & f_n(\hat{T}_n^{(G)})=\frac{2^{\frac{N\hat{T}_n^{(G)}}{B}}\lambda N \ln 2}{2^{\frac{N\hat{T}_n^{(G)}}{B}}-1}, \label{eqn:fn}
\end{align}and $\lambda$ is a constant under which $\sum_{n=1}^N\hat{p}_n=\bar{P}_n$.
\end{proposition}

\begin{proof}
Please refer to Appendix \ref{appendix4}.
\end{proof}

It can be shown that as $\hat{T}_n^{(G)}$'s go to infinity, i.e., the case without fronthaul link constraint in problem (P1), the optimal power allocation given in (\ref{eqn:opt1}) reduces to
\begin{align}\label{eqn:water filling}
\hat{p}_n=\left\{\begin{array}{ll}\frac{1}{\lambda N\ln 2}-\frac{\sigma_n^2}{|h_n|^2}, & {\rm if} ~ \frac{|h_n|^2}{\sigma_n^2}>\lambda N\ln 2, \\ 0, & {\rm otherwise}, \end{array} \right. ~~~ n=1,\cdots,N,
\end{align}which is consistent with the conventional water-filling based power allocation. In the following, we discuss about the impact of fronthaul rate allocation on the optimal power allocation given in (\ref{eqn:opt1}) with finite values of $\hat{T}_n^{(G)}$'s.


\begin{figure}
\begin{center}
\subfigure[A plot of $f_n(\hat{T}_n^{(G)})$ over $\hat{T}_n^{(G)}$]
{\scalebox{0.5}{\includegraphics*{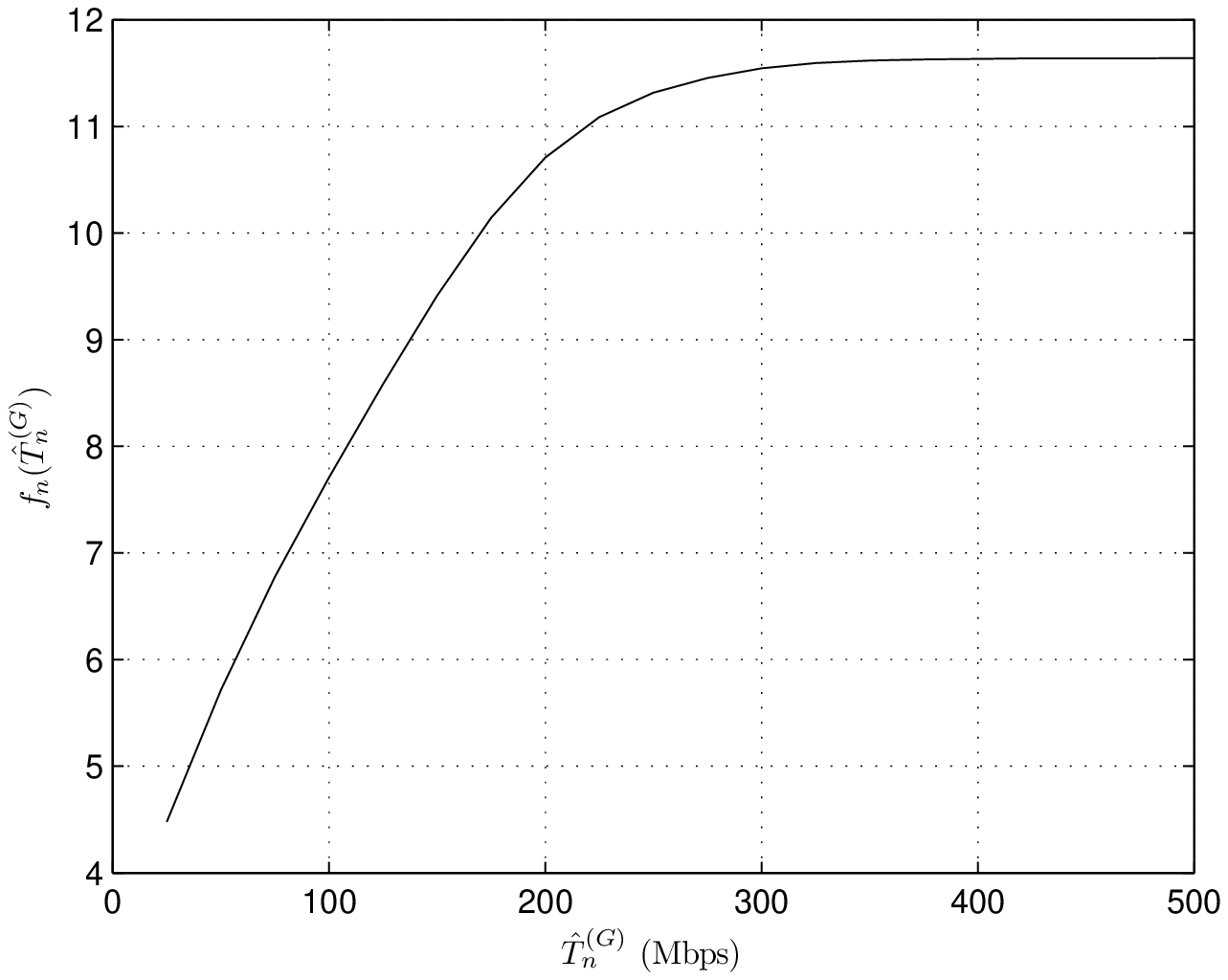}}}
\subfigure[A plot of $\hat{p}_n$'s over $\hat{T}_n^{(G)}$]
{\scalebox{0.5}{\includegraphics*{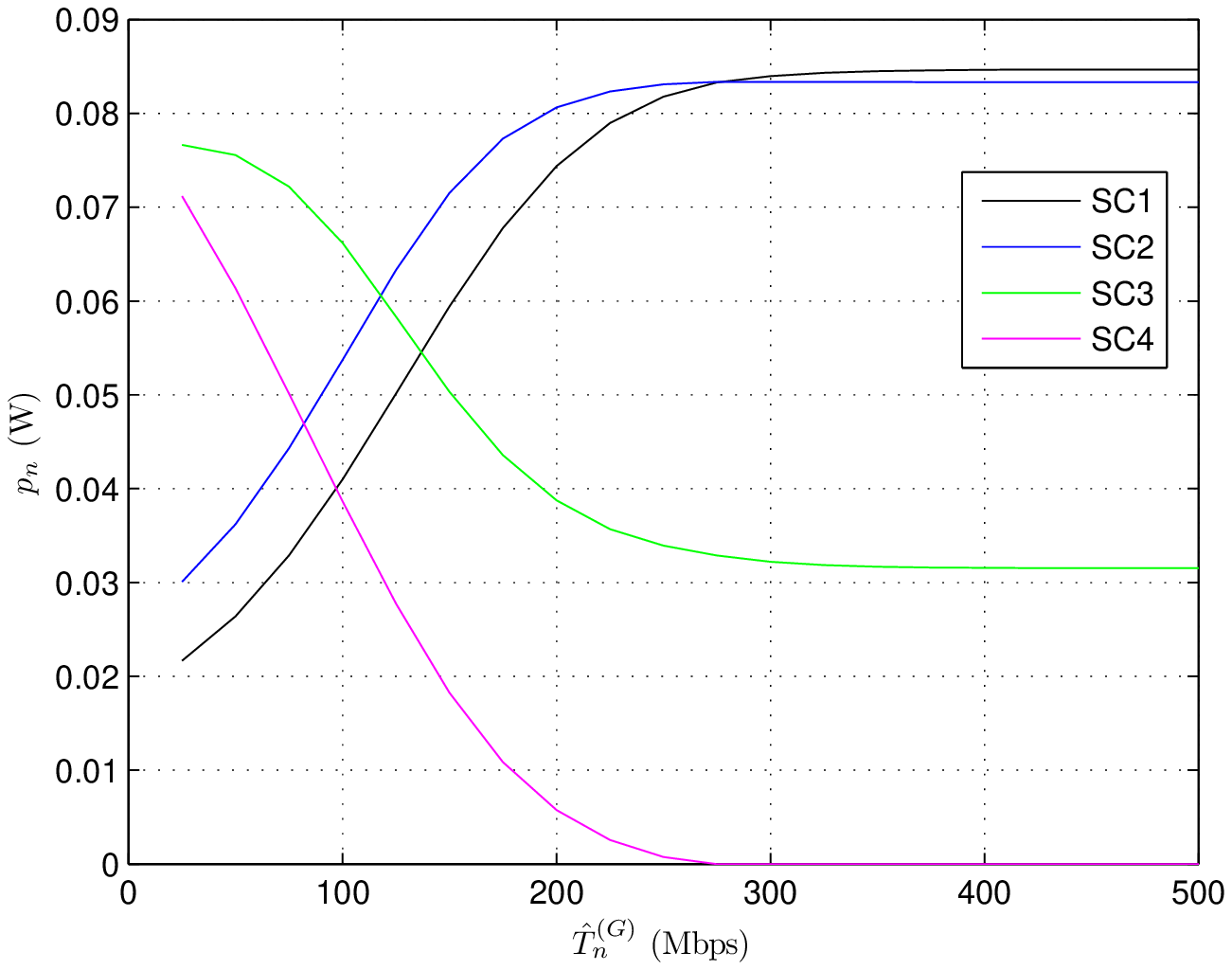}}}
\end{center}\vspace{-15pt}
\caption{Threshold-based power allocation.}\label{fig10} \vspace{-20pt}
\end{figure}

It can be observed from (\ref{eqn:opt1}) that the optimal wireless power allocation with given $\hat{T}_n^{(G)}$'s is threshold-based. In the following, we give a numerical example to investigate the monotonicity of the threshold $f_n(\hat{T}_n^{(G)})$ over $\hat{T}_n^{(G)}$, $\forall n$ (note that in (\ref{eqn:fn}) $\lambda$ is also a function of $\hat{T}_n^{(G)}$'s). In this example, the bandwidth of the wireless link is assumed to be $B=100$MHz, which is equally divided into $4$ SCs. The channel powers are given as $|h_1|^2=1.276\times 10^{-9}$, $|h_2|^2=6.12\times 10^{-10}$, $|h_3|^2=2.9\times 10^{-11}$, $|h_4|^2=1.8\times 10^{-11}$. Moreover, the power spectral density of the background noise is assumed to be $-169$dBm/Hz, and the noise figure due to receiver processing is $7$dB. The transmit power of the user is $23$dBm. It is further assumed that the fronthaul rates are equally allocated among SCs, i.e., $\hat{T}_n^{(G)}=\bar{T}/4$, $\forall n$, and thus $f_n(\hat{T}_n^{(G)})$'s are of the same value. Fig. \ref{fig10} (a) shows the plot of $f_n(\hat{T}_n^{(G)})$ versus $\hat{T}_n^{(G)}$ by increasing the value of $\bar{T}$ in problem (\ref{eqn:p2}). It is observed in this particular setup (and many others used in our simulations for which the results are not shown here due to the space limitation) that in general $f_n(\hat{T}_n^{(G)})$ is increasing with $\hat{T}_n^{(G)}$. This implies that as $\hat{T}_n^{(G)}$ increases, more SCs with weaker channel powers tend to be shut down. The reason is as follows. The dynamic range of the received signal at the SC with stronger channel power is larger, and thus with equal $\hat{T}_n^{(G)}$'s, the corresponding quantization noise level is also larger. When $\hat{T}_n^{(G)}$'s are small, quantization noise dominates the end-to-end rate performance and thus the relatively small quantization noise level at the SC with weaker channel power may offset the loss due to the poor channel condition. However, as $\hat{T}_n^{(G)}$ increases, the quantization noise becomes smaller, until the wireless link dominates the end-to-end performance. In this case, we should shut down some SCs with poor channel conditions just as water-filling based power allocation given in (\ref{eqn:water filling}).


To verify the above analysis, Fig. \ref{fig10} (b) shows the optimal power allocation among the $4$ SCs versus different values of $\hat{T}_n^{(G)}=\bar{T}/4$ in the above numerical example. It is observed that when $\bar{T}_n^{(G)}$ is small, in general the SCs with poorer channel conditions are allocated higher transmit power since the quantization noise levels are small at these SCs. As $\hat{T}_n^{(G)}$ increases, the SCs with poorer channels are allocated less and less transmit power. Specially, when $\hat{T}_n^{(G)}\geq 252.5$Mbps or $\bar{T}\geq 1.1$Gbps, SC 4 with the poorest channel condition is shut down for transmission. It is also observed that when $\hat{T}_n^{(G)}$ is sufficiently large such that the quantization noise is negligible, the power allocation converges to the water-filling based solution given in (\ref{eqn:water filling}).

Next, similar to problem (\ref{eqn:p2}), it can be shown that problem (\ref{eqn:p3}) is a convex problem and thus can be efficiently solved by the Lagrangian duality method. We then have the following proposition.

\begin{proposition}\label{proposition2}
The optimal solution to problem (\ref{eqn:p3}) can be expressed as
\begin{align}\label{eqn:opt2}
\hat{T}_n^{(G)}=\left\{\begin{array}{ll}\frac{B}{N}\log_2\frac{1-\beta B}{\beta B}+\frac{B}{N}\log_2 \nu(n), & {\rm if} ~ \nu_n>\frac{\beta B}{1-\beta B}, \\ 0, & {\rm otherwise},\end{array} \right. ~~~ n=1,\cdots,N,
\end{align}where
\begin{align}
\nu_n=\frac{|h_n|^2\hat{p}_n}{\sigma_n^2},
\end{align}and $\beta<\frac{1}{B}$ is a constant under which $\sum_{n=1}^N\hat{T}_n^{(G)}=\bar{T}$.
\end{proposition}

\begin{proof}
Please refer to Appendix \ref{appendix5}.
\end{proof}

Similar to the optimal power allocation given in (\ref{eqn:opt1}), it can be inferred from Proposition \ref{proposition2} that the optimal fronthaul rate allocation with given $\hat{p}_n$'s is also threshold-based. If the received signal SNR, $\nu_n$, at SC $n$ is below the threshold $\beta B/(1-\beta B)$, the RRH should not quantize and forward the signal at this SC to the BBU for decoding. On the other hand, if $\nu_n>\beta B/(1-\beta B)$, more quantization bits should be allocated to the SCs with higher values of $\nu_n$'s.

After problems (\ref{eqn:p2}) and (\ref{eqn:p3}) are solved by Propositions \ref{proposition1} and \ref{proposition2}, we are ready to propose the overall algorithm to solve problem (P1), which is summarized in Table \ref{table1}. It can be shown that a monotonic convergence can be guaranteed for Algorithm \ref{table1} since the objective value of problem (P1) is increased after each iteration and it is practically bounded.

\begin{table}[htp]
\begin{center}
\caption{\textbf{Algorithm \ref{table1}}: Algorithm for Problem (P1) when $K=1$ and $M=1$} \vspace{0.2cm}
 \hrule
\vspace{0.2cm}
\begin{itemize}
\item[1.] Initialize: Set $T_n^{(G,0)}=\frac{\bar{T}}{N}$, $\forall n$, $R^{(0)}=0$, and $i=0$;
\item[2.] Repeat
\begin{itemize}
\item[a.] $i=i+1$;
\item[b.] Update $\{p_n^{(i)}\}$ by solving problem (\ref{eqn:p2}) with $\hat{T}_{n}^{(G)}=T_n^{(G,i-1)}$, $\forall n$, according to Proposition \ref{proposition1};
\item[c.] Update $\{T_n^{(G,i)}\}$ by solving problem (\ref{eqn:p3}) with $\hat{p}_n=p_n^{(i)}$, $\forall n$, according to Proposition \ref{proposition2};
\end{itemize}
\item[3.] Until $R^{(i)}-R^{(i-1)}\leq \varepsilon$, where $R^{(i)}$ denotes the objective value of problem (P1) achieved by $\{p_n^{(i)}\}$ and $\{T_n^{(G,i)}\}$, and $\varepsilon$ is a small value to control the accuracy of the algorithm.
\end{itemize}
\vspace{0.2cm} \hrule \label{table1}
\end{center} \vspace{-15pt}
\end{table}

\begin{figure}
\begin{center}
\subfigure[Power Allocation]
{\scalebox{0.5}{\includegraphics*{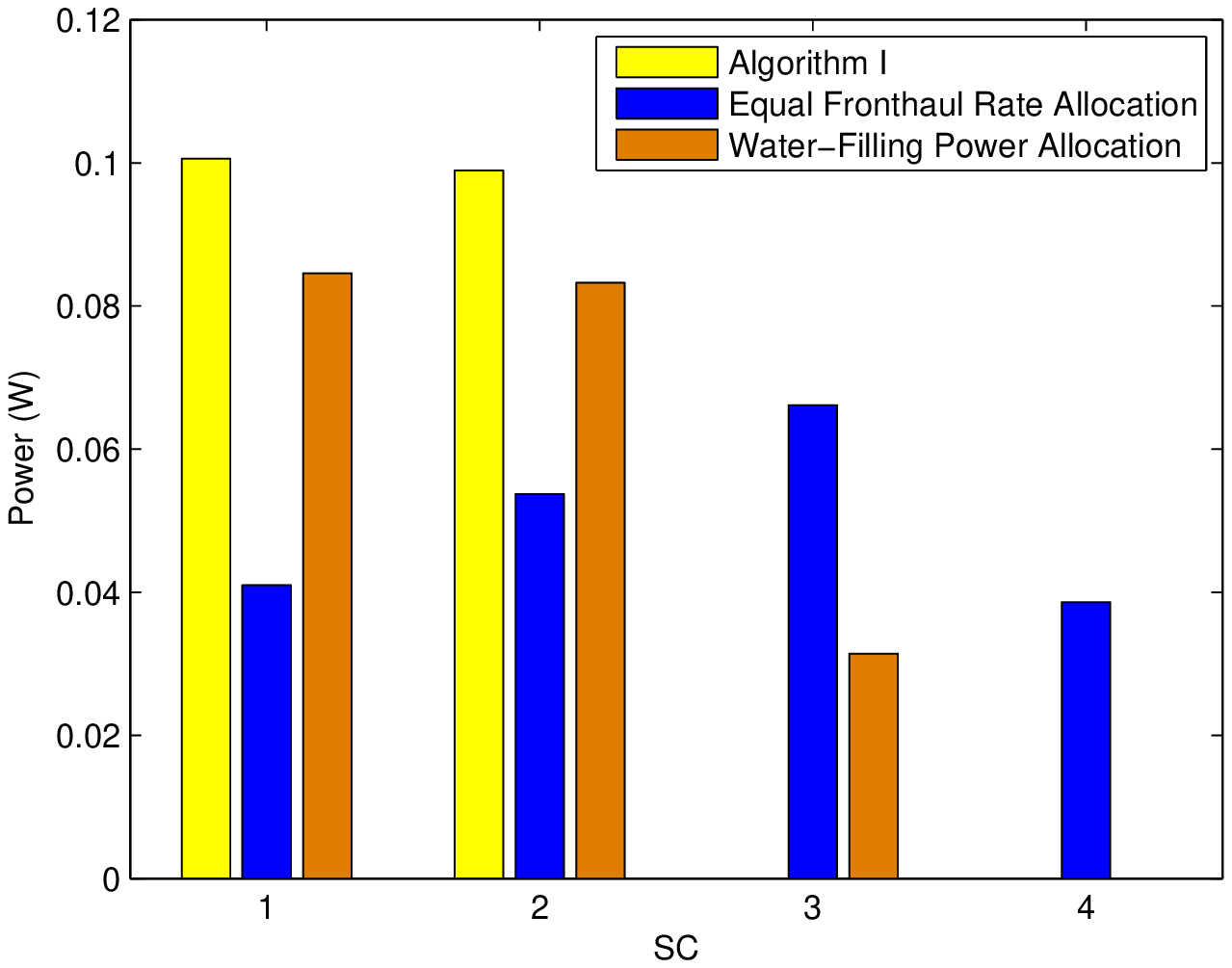}}}
\subfigure[Fronthaul Rate Allocation]
{\scalebox{0.5}{\includegraphics*{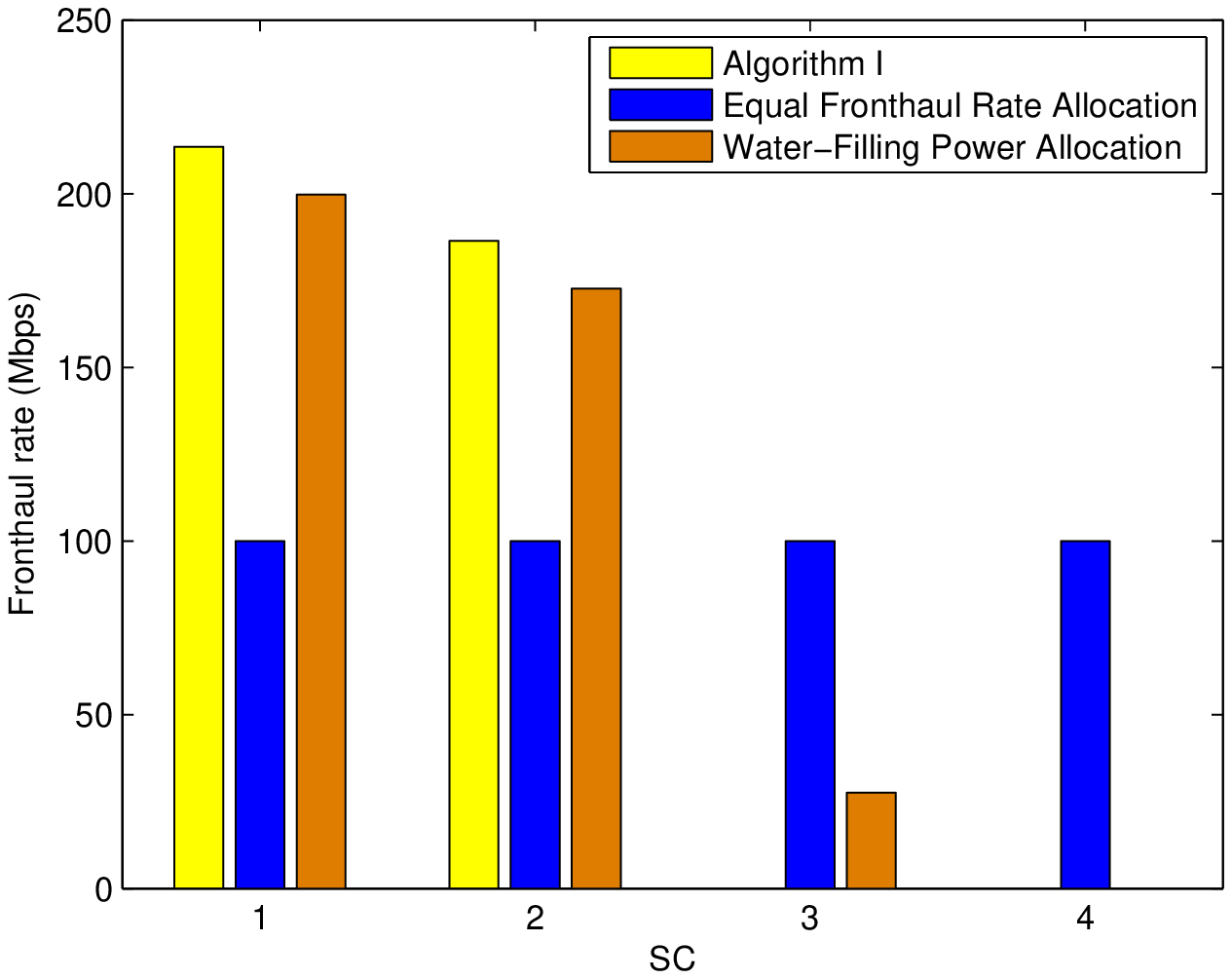}}}
\end{center}\vspace{-15pt}
\caption{Power and fronthaul rate allocation among $4$ SCs.}\label{fig12} \vspace{-20pt}
\end{figure}

With the proposed Algorithm I to solve (P1), we provide a numerical example to analyze the properties of the resulting wireless power and fronthaul rate allocation among SCs. The setup of this example is the same as that for Fig. \ref{fig10}, while the fronthaul link capacity is assumed to be $\bar{T}=400$Mbps. Fig. \ref{fig12} (a) and Fig. \ref{fig12} (b) show the wireless power allocation and the fronthaul rate allocation at each SC, respectively, obtained via Algorithm \ref{table1}. For comparison, in Fig. \ref{fig12} (a) we also provide the power allocation at each SC obtained by solving problem (\ref{eqn:p2}) with equal fronthaul rate allocation, as well as the water-filling based power allocation at each SC (obtained without considering fronthaul link constraint), and in Fig. \ref{fig12} (b) the equal fronthaul rate allocation as well as the fronthaul rate allocation obtained by solving problem (\ref{eqn:p3}) with water-filling based power allocation. It is observed in Fig. \ref{fig12} (a) that Algorithm \ref{table1} results in a more greedy power allocation solution among SCs than the water-filling based method: besides SC $4$, SC $3$ with the second poorest channel condition is also forced to shut down, and the saved power and quantization bits are allocated to SCs $1$ and $2$ with better channel conditions. This is in sharp contrast to the case of equal fronthaul rate allocation for which SC $3$ is allocated the highest transmit power and even SC $4$ with the poorest channel condition is still used for transmission. Moreover, in Fig. \ref{fig12} (b), the fronthaul rate allocations at SCs $1-4$ obtained by Algorithm \ref{table1} are $213.54$Mbps, $186.46$, $0$Mbps, and $0$Mbps, respectively. As a result, different from equal fronthaul rate allocation, Algorithm \ref{table1} tends to allocate more quantization bits to the SCs with strong channel power to explore their good channel conditions, while allocating less (or even no) quantization bits to the SCs with weaker power. A similar fronthaul rate allocation is observed for the water-filling power allocation case.

\subsection{Uniform Scalar Quantization}\label{sec:Uniform Scalar Quantization}

In this subsection, we study problem (P2) in the case of $K=1$ and $M=1$ to evaluate the efficiency of the uniform quantization based scheme. We first solve problem (P2) in this case by extending the results in Section \ref{sec:power control and fronthaul rate allocation for the case of one user and one RRH}. It can be observed that without the last set of constraints involving integer $D_n$'s, problem (P2) is very similar to problem (P1). As a result, in the following we propose a two-stage algorithm to solve problem (P2). First, we ignore the integer constraints in problem (P2), which is denoted by problem (P2-NoInt), and apply an alternating optimization based algorithm similar to Algorithm \ref{table1} to solve it (the details of which are omitted here for brevity). Let $\{\hat{p}_n,\hat{T}_n^{(U)}\}$ denote the converged wireless power and fronthaul rate allocation solution to problem (P2-NoInt). Next, we fix $p_n=\hat{p}_n$'s and find a feasible solution of $T_n^{(U)}$'s based on $\{\hat{p}_n,\hat{T}_n^{(U)}\}$ such that $D_n=NT_n^{(U)}/2B$'s are integers, $\forall n$, in problem (P2). This is achieved by rounding each $N\hat{T}_n^{(U)}/2B$ to its nearby integer as follows:
\begin{align}\label{eqn:feasible}
\frac{NT_n^{(U)}}{2B}=\left\{\begin{array}{ll}\lfloor \frac{N\hat{T}_n^{(U)}}{2B} \rfloor, & {\rm if} ~ \frac{N\hat{T}_n^{(U)}}{2B}-\lfloor \frac{N\hat{T}_n^{(U)}}{2B} \rfloor \leq \alpha, \\ \lceil \frac{N\hat{T}_n^{(U)}}{2B} \rceil, & {\rm otherwise},\end{array}\right. ~~~ n=1,\cdots,N,
\end{align}where $0\leq \alpha \leq 1$, and $\lfloor x \rfloor$ denotes the maximum integer that is no larger than $x$. Note that we can always find a feasible solution of $T_n$'s by simply setting $\alpha=1$ in (\ref{eqn:feasible}) since in this case we have $\sum_{n=1}^N T_n^{(U)}\leq \sum_{n=1}^N \hat{T}_n^{(U)} \leq \bar{T}$. In the following, we show how to find a better feasible solution by optimizing $\alpha$. It can be observed from (\ref{eqn:feasible}) that with decreasing $\alpha$, the values of $T_n^{(U)}$'s will be non-decreasing, $\forall n$. As a result, the objective value of problem (P2) will be non-decreasing, but the fronthaul link constraint in problem (P2) will be more difficult to satisfy. Thereby, we propose to apply a simple bisection method to find the optimal value of $\alpha$, denoted by $\alpha^\ast$, which is summarized in Table \ref{table4}. After $\alpha^\ast$ is obtained, the feasible solution of $T_n^{(U)}$'s can be efficiently obtained by taking $\alpha^\ast$ into (\ref{eqn:feasible}). Notice that by (\ref{eqn:feasible}) the number of quantization bits per SC, $D_n$, is now allowed to be zero, instead of being a strictly positive integer as assumed in Sections \ref{sec:Two Scalar Quantization Models} and \ref{sec:Problem Formulation}.\footnote{In the case of $D_n=0$ and hence $T_n^{(U)}=0$, for any SC $n$, the achievable end-to-end rate for the uniform scalar quantization given in (\ref{eqn:new end to end rate}) no longer holds, which instead should be set to zero intuitively.}

\begin{table}[htp]
\begin{center}
\caption{\textbf{Algorithm \ref{table4}}: Algorithm to Find Feasible Solution of $T_n^{(U)}$'s to problem (P2)} \vspace{0.2cm}
 \hrule
\vspace{0.2cm}

\begin{itemize}
\item[1.] Initialize $\alpha_{{\rm min}}=0$, $\alpha_{{\rm max}}=1$;
\item[2.] Repeat
\begin{itemize}
\item[a.] Set $\alpha=\frac{\alpha_{{\rm min}}+\alpha_{{\rm max}}}{2}$;
\item[b.] Take $\alpha$ into (\ref{eqn:feasible}). If $T_n^{(U)}$'s, $\forall n$, satisfy the fronthaul link capacity constraint in problem (P2), set $\alpha_{{\rm max}}=\alpha$; otherwise, set $\alpha_{{\rm min}}=\alpha$;
\end{itemize}
\item[3.] Until $\alpha_{{\rm max}}-\alpha_{{\rm min}}<\varepsilon$, where $\varepsilon$ is a small value to control the accuracy of the algorithm;
\item[4.] Take $\alpha$ into (\ref{eqn:feasible}) to obtain the feasible solution of $T_n^{(U)}$'s, $\forall n$.
\end{itemize}

\vspace{0.2cm} \hrule \label{table4}
\end{center} \vspace{-15pt}
\end{table}

Next, we evaluate the end-to-end rate performance of the uniform scalar quantization based scheme in the case of $K=1$ and $M=1$. Note that a cut-set based capacity upper bound of our studied C-RAN is \cite{Shamai09}
\begin{align}\label{eqn:capacity upper bound}
C=\min\left(\frac{1}{N}\sum\limits_{n=1}^N\log_2\left(1+\frac{|h_n|^2p_n^{{\rm wf}}}{\sigma_n^2}\right),\frac{\bar{T}}{B}\right) ~ {\rm bps/Hz},
\end{align}where $\{p_n^{{\rm wf}}\}$ is the water-filling based optimal power solution given in (\ref{eqn:water filling}).
\begin{proposition}\label{proposition3}
In the case of $K=1$ and $M=1$, let $\bar{R}_{{\rm sum}}^{(G)}$ denote the optimal value of problem (P1) with an additional set of constraints of
\begin{align}\label{eqn:constraint1}
q_n=\frac{|h_n|^2p_n+\sigma_n^2}{2^{\frac{NT_n^{(G)}}{B}}-1}=\sigma_n^2, ~~~ n=1,\cdots,N.
\end{align}Then we have $\bar{R}_{{\rm sum}}^{(G)}/B\geq C-1$.
\end{proposition}

\begin{proof}
Please refer to Appendix \ref{appendix6}.
\end{proof}

Proposition \ref{proposition3} implies that with the simple solution $\{p_n=\check{p}_n,T_n^{(G)}=(B/N)\log_2(2+|h_n|^2\check{p}_n/\sigma_n^2)\}$ with $\check{p}_n$'s denoting the optimal solution to problem (\ref{eqn:fix quantization in P1}) given in Appendix \ref{appendix6}, the Gaussian test channel based scheme can achieve a capacity to within $1$bps/Hz. Next, for the uniform scalar quantization, by setting the quantization noise level given in (\ref{eqn:quantization noise power}) as $q_n=3\sigma_n^2$, $\forall n$, in problem (P2-NoInt), we have the following proposition.

\begin{proposition}\label{proposition4}
In the case of $K=1$ and $M=1$, $\{p_n=\check{p}_n,T_n^{(G)}=(B/N)\log_2(1+|h_n|^2\check{p}_n/\sigma_n^2)\}$ is a feasible solution to problem (P2-NoInt). Let $\bar{R}_{{\rm sum}}^{(U)}$ denote the objective value of problem (P2-NoInt) achieved by the above solution, we then have $\bar{R}_{{\rm sum}}^{(U)}/B> \bar{R}_{{\rm sum}}^{(G)}/B-1$.
\end{proposition}

\begin{proof}
Please refer to Appendix \ref{appendix7}.
\end{proof}

It can be inferred from Propositions \ref{proposition3} and \ref{proposition4} that $\bar{R}_{{\rm sum}}^{(U)}/B> \bar{R}_{{\rm sum}}^{(G)}/B-1\geq C-2$. As a result, we have the following corollary.

\begin{corollary}\label{corollary1}
Without the constraints that the number of quantization bits per SC is an integer, with the simple solution $\{p_n=\check{p}_n,T_n^{(G)}=(B/N)\log_2(1+|h_n|^2\check{p}_n/\sigma_n^2)\}$, the uniform scalar quantization based scheme at least achieves a capacity to within $2$bps/Hz in the case of $K=1$ and $M=1$.
\end{corollary}

Corollary \ref{corollary1} gives a worst-case performance gap of the proposed uniform quantization based scheme to the cut-set upper bound $C$ in (\ref{eqn:capacity upper bound}) if we ignore the constraints that each quantization level is represented by an integer number of bits. However, it is difficult to analyze the performance loss due to these integer constraints. In the following subsection, we will provide a numerical example to show the impact of the integer constraints on the end-to-end rate performance.

\subsection{Numerical Example}\label{sec:Numerical Example}


In this subsection, we provide a numerical example to verify our results for the case of $K=1$ and $M=1$. The setup of this example is summarized as follows. The channel bandwidth is assumed to be $B=100$MHz, which is equally divided into $N=32$ SCs. The user's transmit power is $23$dBm. It is assumed that the distance between the user and the RRH is $d=50$m. The pass loss model is $L=30.6+36.7\log_{10}(d)$ dB. Moreover, it is assumed that the power spectral density of the AWGN at the RRH is $-169$dBm/Hz, and the noise figure is $7$dB. First, we evaluate the performance of the proposed uniform scalar quantization based scheme against that of the Gaussian test channel based scheme as well as the capacity upper bound given in (\ref{eqn:capacity upper bound}). Fig. \ref{fig9} shows the end-to-end rate achieved by various schemes versus the fronthaul link capacity. Note that with the algorithm proposed for problem (P2-NoInt) in Section \ref{sec:Uniform Scalar Quantization}, we use $\{p_n=\check{p}_n,T_n^{(G)}=(B/N)\log_2(1+|h_n|^2\check{p}_n/\sigma_n^2)\}$ as the initial point such that the worst-case performance gap shown in Corollary \ref{corollary1} can be guaranteed. It is observed from Fig. \ref{fig9} that for various values of $\bar{T}$, uniform scalar quantization based scheme without the integer constraints in problem (P2) does achieve a capacity within $2$bps/Hz to $C$. Moreover, it is observed that with Algorithm \ref{table4}, the performance loss due to the integer constraints is negligible. However, if we simply set $\alpha=1$ in (\ref{eqn:feasible}) to find feasible $T_n^{(U)}$'s, there will be a considerable rate loss. As a result, our proposed Algorithm \ref{table4} is practically useful for setting $\alpha$ such that uniform scalar quantization based scheme can perform very close to the capacity upper bound. Last, it is observed that the performance gap of all the schemes to the upper bound $C$ vanishes as the fronthaul link capacity increases. This is because if $\bar{T}$ is sufficiently large at the RRH, each symbol can be quantized by a large number of bits such that the specific quantization method does not affect the quantization noise significantly.

\begin{figure}
\begin{center}
\scalebox{0.5}{\includegraphics*{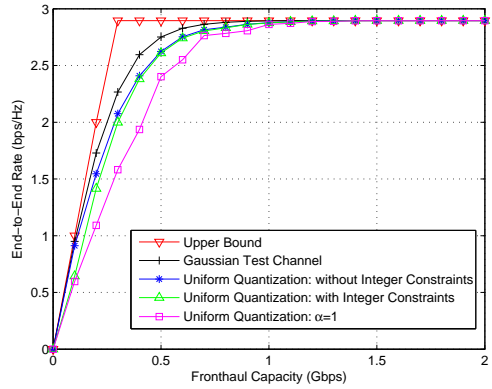}}
\end{center}\vspace{-15pt}
\caption{Performance of uniform scalar quantization.}\label{fig9} \vspace{-20pt}
\end{figure}

To further illustrate the gain from joint optimization of wireless power and fronthaul rate allocation, in the following we introduce some benchmark schemes where either wireless power or fronthaul rate allocation is optimized, but not both.

\begin{itemize}
\item{\bf Benchmark Scheme 1: Equal Power Allocation.} In this scheme, the user allocates its transmit power equally to each SC, i.e., $p_n=\bar{P}/N$, $\forall n$. Then, with the given equal power allocation, we optimize the fronthaul rate allocation at the RRH to maximize the end-to-end rate.
\item{\bf Benchmark Scheme 2: Water-Filling Power Allocation.} In this scheme, the user ignores the fronthaul link constraints and allocates its transmit power based on water-filling solution as shown in (\ref{eqn:water filling}). Then, with the given water-filling based power allocation, we optimize the fronthaul rate allocation at the RRH to maximize the end-to-end rate.
\item{\bf Benchmark Scheme 3: Equal Fronthaul Rate Allocation.} In this scheme, the RRH equally allocates its fronthaul link capacity among SCs, $T_n^{(U)}=\bar{T}/N$. Then, with the given equal fronthaul rate allocation, we optimize the transmit power of the user to maximize the end-to-end rate.
\item{\bf Benchmark Scheme 4: Equal Power and Fronthaul Rate Allocation.} In this scheme, the user allocates its transmit power equally to each SC, and the RRH equally allocates its fronthaul link bandwidth among SCs.
\end{itemize}

\begin{figure}
\begin{center}
\scalebox{0.5}{\includegraphics*{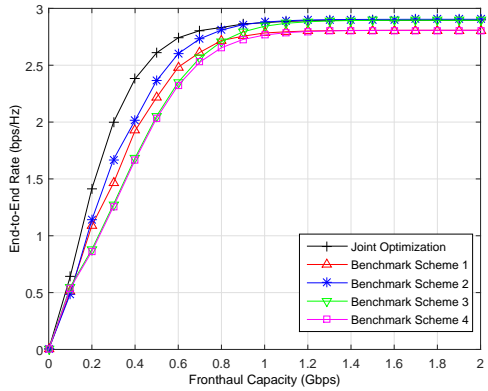}}
\end{center}\vspace{-15pt}
\caption{Performance gain due to joint optimization of wireless power and fronthaul rate allocation.}\label{fig4} \vspace{-20pt}
\end{figure}

Fig. \ref{fig4} shows the performance comparison among various proposed solutions for the uniform scalar quantization based scheme. It is observed that compared with Benchmark Schemes 1-4 where only either wireless power or fronthaul rate allocation is optimized, our joint optimization solution proposed in Section \ref{sec:Uniform Scalar Quantization} achieves a much higher end-to-end rate, especially when the fronthaul link capacity is small, e.g., $\bar{T}\leq 0.5$ Gbps. Furthermore, it is observed from Benchmark Schemes 1 and 3 that when $\bar{T}$ is small, fronthaul rate optimization plays the dominant role in improving the end-to-end rate performance, while when $\bar{T}$ is large, most of the optimization gain comes from the wireless power allocation. Furthermore, when $\bar{T}$ is sufficiently large, the performance of Benchmark Schemes 2 and 3, for which wireless power allocation is optimized, even converges to the joint optimization solution proposed in Section \ref{sec:Uniform Scalar Quantization}.

\section{General Case: Multiple Users and Multiple RRHs}\label{sec:General Case: Multiple Users and Multiple RRHs}

In this section, we consider the joint wireless power allocation and fronthaul rate allocation in the general C-RAN with multiple users and multiple RRHs, i.e., $K\geq 1$ and $M\geq 1$.

\subsection{Gaussian Test Channel}\label{sec:power control and fronthaul rate allocation for the case of multiple users and multiple RRHs}

In this subsection, we solve problem (P1). It is worth noting that different from Section \ref{sec:power control and fronthaul rate allocation for the case of one user and one RRH}, in the case of multiple RRHs, the throughput $R_{{\rm sum}}^{(G)}$ given in (\ref{eqn:test channel rate}) is not concave over $T_{m,n}^{(G)}$'s with given $p_{k,n}$'s due to the summation over $m$ in (\ref{eqn:optimal SINR}). As a result, the alternating optimization based solution proposed in Section \ref{sec:power control and fronthaul rate allocation for the case of one user and one RRH} cannot be directly extended to the general case of $K\geq 1$ and $M\geq 1$.

To deal with the above difficulty, we change the design variables in problem (P1). Define
\begin{align}\label{eqn:new fronthaul rate SC}
\psi_{m,n}=2^{\frac{NT_{m,n}^{(G)}}{B}}-1, ~~~ \forall m,n.
\end{align}Then, by changing the design variables of problem (P1) from $\{p_{k,n},T_{m,n}^{(G)}\}$ to $\{p_{k,n},\psi_{m,n}\}$, problem (P1) is transformed into the following problem.
\begin{align}\begin{small}\mathop{\mathtt{Maximize}}_{\{p_{k,n},\psi_{m,n}\}} \end{small} & ~~~ \begin{small} \frac{B}{N}\sum\limits_{k=1}^K\sum\limits_{n\in \Omega_k}\log_2\left(1+\sum\limits_{m=1}^M\frac{|h_{m,k,n}|^2p_{k,n}\psi_{m,n}}{\sigma_{m,n}^2\psi_{m,n}+|h_{m,k,n}|^2p_{k,n}+\sigma_{m,n}^2}\right) \end{small} \nonumber \\
\begin{small} \mathtt {Subject \ to} \end{small} & ~~~ \begin{small} \frac{B}{N}\sum\limits_{n=1}^N\log_2(1+\psi_{m,n})\leq \bar{T}_m, ~~~ \forall m, \end{small} \nonumber \\ & ~~~ \begin{small} \sum\limits_{n=1}^Np_{k,n}\leq \bar{P}_k, ~~~ \forall k. \end{small} \label{eqn:p4}
\end{align}Problem (\ref{eqn:p4}) is still a non-convex problem. In the following, we propose to apply the techniques of alternating optimization as well as convex approximation to solve it.

First, by fixing $\psi_{m,n}=\hat{\psi}_{m,n}$'s, we optimize the transmit power allocation $p_{k,n}$'s by solving the following problem.
\begin{align}\begin{small}\mathop{\mathtt{Maximize}}_{\{p_{k,n}\}}\end{small} & ~~~ \begin{small}\frac{B}{N}\sum\limits_{k=1}^K\sum\limits_{n\in \Omega_k}\log_2\left(1+\sum\limits_{m=1}^M\frac{|h_{m,k,n}|^2p_{k,n}\hat{\psi}_{m,n}}{\sigma_{m,n}^2\hat{\psi}_{m,n}+|h_{m,k,n}|^2p_{k,n}+\sigma_{m,n}^2}\right)\end{small} \nonumber \\
\begin{small}\mathtt {Subject \ to}\end{small} & ~~~ \begin{small}\sum\limits_{n=1}^Np_{k,n}\leq \bar{P}_k, ~~~ \forall k.\end{small} \label{eqn:p5}
\end{align}Let $\hat{p}_{k,n}$'s denote the optimal solution to problem (\ref{eqn:p5}). Then, by fixing $p_{k,n}=\hat{p}_{k,n}$'s, we optimize the fronthaul rate allocation by solving the following problem.
\begin{align}\mathop{\mathtt{Maximize}}_{\{\psi_{m,n}\}} & ~~~ \frac{B}{N}\sum\limits_{k=1}^K\sum\limits_{n\in \Omega_k}\log_2\left(1+\sum\limits_{m=1}^M\frac{|h_{m,k,n}|^2\hat{p}_{k,n}\psi_{m,n}}{\sigma_{m,n}^2\psi_{m,n}+|h_{m,k,n}|^2\hat{p}_{k,n}+\sigma_{m,n}^2}\right) \nonumber \\
\mathtt {Subject \ to} & ~~~ \frac{B}{N}\sum\limits_{n=1}^N\log_2(1+\psi_{m,n})\leq \bar{T}_m, ~~~ \forall m. \label{eqn:p6}
\end{align}Let $\hat{\psi}_{m,n}$'s denote the optimal solution to problem (\ref{eqn:p6}). Then, the above update of $p_{k,n}$'s and $\psi_{m,n}$'s is iterated until convergence. In the following, we provide how to solve problems (\ref{eqn:p5}) and (\ref{eqn:p6}), respectively.

First, we consider problem (\ref{eqn:p5}). We have the following lemma.
\begin{lemma}\label{lemma1}
The objective function of problem (\ref{eqn:p5}) is a concave function over $\{p_{k,n}\}$.
\end{lemma}

\begin{proof}
Please refer to Appendix \ref{appendix8}.
\end{proof}

According to Lemma \ref{lemma1}, problem (\ref{eqn:p5}) is a convex optimization problem. As a result, its optimal solution can be efficiently obtained via the interior-point method \cite{Boyd04}.

Next, we consider problem (\ref{eqn:p6}). Similar to Lemma \ref{lemma1}, it can be shown that the objective function of problem (\ref{eqn:p6}) is a concave function over $\psi_{m,n}$'s. However, the fronthaul link capacity constraints in problem (\ref{eqn:p6}) are not convex. In the following, we apply the convex approximation technique to convexify the fronthaul link capacity constraints. Specifically, since according to (\ref{eqn:new fronthaul rate SC}) $T_m^{(G)}$ is concave over $\psi_{m,n}$'s, its first-order approximation serves as an upper bound to it, i.e.,
\begin{align}\label{eqn:first order approximation}
T_m^{(G)}=\frac{B}{N}\sum\limits_{n=1}^N\log_2(1+\psi_{m,n})\leq \frac{B}{N}\sum\limits_{n=1}^N\left(\log_2(1+\tilde{\psi}_{m,n})+\frac{\psi_{m,n}-\tilde{\psi}_{m,n}}{(1+\tilde{\psi}_{m,n})\ln2}\right), ~~~ m=1,\cdots,M.
\end{align}Note that the above inequality holds given any $\tilde{\psi}_{k,n}$'s. As a result, we solve the following problem via a relaxation of problem (\ref{eqn:p6}).
\begin{align}\mathop{\mathtt{Maximize}}_{\{\psi_{m,n}\}} & ~~~ \frac{B}{N}\sum\limits_{k=1}^K\sum\limits_{n\in \Omega_k}\log_2\left(1+\sum\limits_{m=1}^M\frac{|h_{m,k,n}|^2\hat{p}_{k,n}\psi_{m,n}}{\sigma_{m,n}^2\psi_{m,n}+|h_{m,k,n}|^2\hat{p}_{k,n}+\sigma_{m,n}^2}\right) \nonumber \\
\mathtt {Subject \ to} & ~~~ \frac{B}{N}\sum\limits_{n=1}^N\left(\log_2(1+\tilde{\psi}_{m,n})+\frac{\psi_{m,n}-\tilde{\psi}_{m,n}}{(1+\tilde{\psi}_{m,n})\ln2}\right)\leq \bar{T}_m, ~~~ \forall m. \label{eqn:p7}
\end{align}Problem (\ref{eqn:p7}) is a convex problem, and thus its optimal solution, denoted by $\check{\psi}_{m,n}$'s, can be efficiently obtained via the interior-point method \cite{Boyd04}. Then we have the following lemma.

\begin{lemma}\label{lemma2}
Suppose that $\tilde{\psi}_{m,n}$'s is a feasible solution to problem (\ref{eqn:p6}), i.e., $\frac{B}{N}\sum_{n=1}^N\log_2(1+\tilde{\psi}_{m,n})\leq \bar{T}_m$, $\forall m$. Then, $\check{\psi}_{m,n}$'s is a feasible solution to problem (\ref{eqn:p6}) and achieves an objective value no smaller than that achieved by the solution $\tilde{\psi}_{m,n}$'s.
\end{lemma}

\begin{proof}
Please refer to Appendix \ref{appendix9}.
\end{proof}

Since the optimal solution to problem (\ref{eqn:p6}), i.e., $\hat{\psi}_{m,n}$'s, is difficult to obtain, in the following we use $\check{\psi}_{m,n}$ as the solution to (\ref{eqn:p6}) according to Lemma \ref{lemma2}, i.e., $\hat{\psi}_{m,n}=\check{\psi}_{m,n}$, $\forall m,n$.

After problems (\ref{eqn:p5}) and (\ref{eqn:p6}) are solved, we are ready to propose the overall iterative algorithm to solve problem (\ref{eqn:p4}), which is summarized in Table \ref{table3}. Note that in Step 2.c., we set $\tilde{\psi}_{m,n}=\psi_{m,n}^{(i-1)}$'s in problem (\ref{eqn:p7}). According to Lemma \ref{lemma2}, $\psi_{m,n}^{(i)}$'s will achieve a sum-rate that is no smaller than that achieved by $\psi_{m,n}^{(i-1)}$'s. To summarize, a monotonic convergence can be guaranteed for Algorithm III since the objective value of problem (\ref{eqn:p4}) is increased after each iteration and it is upper-bounded by a finite value.

\begin{table}[htp]
\begin{center}
\caption{\textbf{Algorithm III}: Algorithm for Problem (\ref{eqn:p4})} \vspace{0.2cm}
 \hrule
\vspace{0.2cm}
\begin{itemize}
\item[1.] Initialize: Set $\psi_{m,n}^{(0)}=2^{\frac{\bar{T}_m}{B}}-1$, $\forall m,n$, $R^{(0)}=0$, and $i=0$;
\item[2.] Repeat
\begin{itemize}
\item[a.] $i=i+1$;
\item[b.] Update $\{p_{k,n}^{(i)}\}$ by solving problem (\ref{eqn:p5}) with $\hat{\psi}_{m,n}=\psi_{m,n}^{(i-1)}$, $\forall m,n$, via interior-point method;
\item[c.] Update $\{\psi_{m,n}^{(i)}\}$ by solving problem (\ref{eqn:p7}) with $\hat{p}_{k,n}=p_{k,n}^{(i)}$ and $\tilde{\psi}_{m,n}=\psi_{m,n}^{(i-1)}$, $\forall m,n$, via interior-point method;
\end{itemize}
\item[3.] Until $R^{(i)}-R^{(i-1)}\leq \varepsilon$, where $R^{(i)}$ denotes the objective value of problem (\ref{eqn:p4}) achieved by the solution $\{p_{k,n}^{(i)},\psi_{m,n}^{(i)}\}$, and $\varepsilon$ is a small value to control the accuracy of the algorithm.
\end{itemize}
\vspace{0.2cm} \hrule \label{table3}
\end{center}\vspace{-20pt}
\end{table}

\subsection{Uniform Scalar Quantization}\label{sec:Joint Optimization of Wireless Power Allocation and Fronthaul Rate Allocation}

In this subsection, we propose an efficient algorithm to solve problem (P2) by jointly optimizing the wireless power allocation as well as the fronthaul rate allocation. To be consistent with the solution to problem (P1) proposed in Section \ref{sec:power control and fronthaul rate allocation for the case of multiple users and multiple RRHs}, we define
\begin{align}
\psi_{m,n}=2^{\frac{NT_{m,n}^{(U)}}{B}}=2^{2D_{m,n}}, ~~~ \forall m,n.
\end{align}Then, by changing the design variables from $\{p_{k,n},T_{m,n}^{(U)}\}$ into $\{p_{k,n},\psi_{m,n}\}$, problem (P2) is transformed into the following problem.
\begin{align}\mathop{\mathtt{Maximize}}_{\{p_{k,n},\psi_{m,n}\}} & ~~~ \frac{B}{N}\sum\limits_{k=1}^K\sum\limits_{n\in \Omega_k}\log_2\left(1+\sum\limits_{m=1}^M\frac{|h_{m,k,n}|^2p_{k,n}\psi_{m,n}}{\sigma_{m,n}^2\psi_{m,n}+3(|h_{m,k,n}|^2p_{k,n}+\sigma_{m,n}^2)}\right) \nonumber \\
\mathtt {Subject \ to} & ~~~ \frac{B}{N}\sum\limits_{n=1}^N \log_2\psi_{m,n}\leq \bar{T}_m, ~~~ \forall m, \nonumber \\ & ~~~ \sum\limits_{n\in \Omega_k} p_{k,n}\leq \bar{P}_k, ~~~ \forall k, \nonumber \\ & ~~~ \psi_{m,n}=2^{2D_{m,n}}, ~ D_{m,n}\in \{1,2,\cdots\} ~ {\rm is \ an \ integer}, ~ \forall m,n. \label{eqn:p8}
\end{align}

It can be observed that if we ignore the last set of constraints involving integers $D_{m,n}$'s, then problem (\ref{eqn:p8}) is very similar to problem (\ref{eqn:p4}). As a result, we propose a two-stage algorithm to solve problem (\ref{eqn:p8}). First, we ignore the last constraints in problem (\ref{eqn:p8}) and apply an alternating optimization based algorithm similar to Algorithm III to solve it (the details of which are omitted here for brevity). Let $\{\hat{p}_{k,n},\hat{\psi}_{m,n}\}$ denote the obtained solution. Then we fix $p_{k,n}=\hat{p}_{k,n}$'s and find a feasible solution $\psi_{m,n}$'s based on $\hat{\psi}_{m,n}$'s such that $D_{m,n}=\frac{1}{2}\log_2\psi_{m,n}$'s are integers. For any given $m=\bar{m}$, this is done by rounding $\frac{1}{2}\log_2\hat{\psi}_{\bar{m},n}$'s, $\forall n$, to their nearby integers as follows:
\begin{align}\label{eqn:feasible1}
\frac{1}{2}\log_2\psi_{\bar{m},n}=\left\{\begin{array}{ll}\lfloor \frac{1}{2}\log_2\hat{\psi}_{\bar{m},n} \rfloor, & {\rm if} ~ \frac{1}{2}\log_2\hat{\psi}_{\bar{m},n}-\lfloor \frac{1}{2}\log_2\hat{\psi}_{\bar{m},n} \rfloor \leq \alpha_{\bar{m}}, \\ \lceil \frac{1}{2}\log_2\hat{\psi}_{\bar{m},n} \rceil, & {\rm otherwise},\end{array}\right. ~~~ n=1,\cdots,N,
\end{align}where $0\leq \alpha_{\bar{m}} \leq 1$, $\forall \bar{m}$. Similar to Algorithm \ref{table4} for the special case of $K=1$ and $M=1$, the optimal value of $\alpha_{\bar{m}}$ can be efficiently obtained via a simple bisection method, and thus a feasible solution of $\psi_{\bar{m},n}$'s, $\forall n$, is obtained according to (\ref{eqn:feasible1}). Last, by searching $\bar{m}$ from $1$ to $M$, the overall feasible solution $\{\psi_{m,n}\}$ is obtained.

\subsection{Numerical Example}\label{sec:Numerical Results}

In this subsection, we provide a numerical example to evaluate the sum-rate performance of the proposed uniform scalar quantization based scheme in a single C-RAN cluster with $M=7$ RRHs and $K=16$ users randomly distributed in a circular area of radius $100$m. It is assumed that the $B=300$MHz bandwidth of the wireless link is equally divided into $N=64$ SCs, and each user is pre-allocated $N/K=4$ SCs. It is further assumed that the capacities of all the fronthaul links are identical, i.e., $\bar{T}_m=T$, $\forall m$. The other setup parameters are the same as those used in Section \ref{sec:Numerical Example}. Similar to the single-user single-RRH case in Section \ref{sec:Numerical Example}, we provide various benchmark schemes. Note that Benchmark Schemes 1-4 introduced in Section \ref{sec:Numerical Example} can be simply extended to the general case of $K\geq 1$ and $M\geq 1$. Furthermore, to compare the sum-rate performance between our studied OFDMA-based C-RAN and conventional OFDMA-based cellular networks, we also consider the following benchmark scheme.

\begin{itemize}
\item{\bf Benchmark Scheme 5: Conventional OFDMA.} In this scheme, we assume that each RRH operates like conventional BS in cellular networks which directly decodes the messages of its served users, rather than forwarding its received signals to the BBU for a joint decoding. For simplicity, we assume that each user is served by its nearest RRH. Then, the optimal power solution for each user $k$ among its assigned SCs $\Omega_k$ is the standard ``water-filling'' solution given in (\ref{eqn:water filling}).
\end{itemize}

\begin{figure}
\begin{center}
\scalebox{0.5}{\includegraphics*{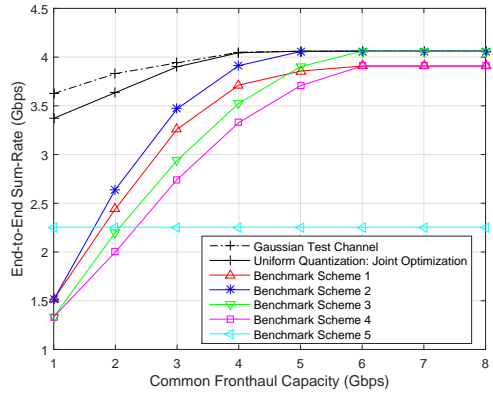}}
\end{center}\vspace{-20pt}
\caption{Uniform scalar quantization: end-to-end sum-rate versus common fronthaul link capacity.}\label{fig14} \vspace{-25pt}
\end{figure}

Fig. \ref{fig14} shows the end-to-end sum-rate performance versus the common fronthaul link capacity, $T$, achieved by uniform quantization, Gaussian test channel, as well as Benchmark Schemes 1-5 (Note that in Benchmark Scheme 5, since each RRH decodes the messages locally, we assume that the sum-rate is a constant regardless of fronthaul capacities). It is observed that with our proposed algorithm in Section \ref{sec:Joint Optimization of Wireless Power Allocation and Fronthaul Rate Allocation}, the sum-rate achieved by the uniform scalar quantization based scheme is very close to that achieved by the Gaussian test channel based scheme for various fronthaul capacities. Furthermore, this performance gap vanishes as the fronthaul link capacities increase at all RRHs. It is also observed that compared with Benchmark Schemes 1-4 where only either wireless power or fronthaul rate allocation is optimized, our joint optimization solution proposed in Section \ref{sec:Joint Optimization of Wireless Power Allocation and Fronthaul Rate Allocation} achieves a much higher sum-rate, especially when the fronthaul link capacities are not sufficiently high. By comparing with Fig. \ref{fig4}, it is observed that the joint optimization gain is more significant over the case of single user and single RRH. Last, it is observed that with joint optimization of wireless and fronthaul resource allocation, the sum-rate achieved by proposed OFDMA-based C-RAN is much higher than that achieved by Benchmark Scheme 5, i.e., conventional OFDMA, under the moderate capacity of current commercial fronthaul such as several Gbps.

\section{Conclusions and Future Work}\label{sec:Conclusion}

In this paper, we have proposed joint wireless power control and fronthaul rate allocation optimization to maximize the throughput performance of an OFDMA-based broadband C-RAN system. In particular, we have considered using practical uniform scalar quantization instead of the information-theoretical quantization method in the system design. Efficient algorithms have been proposed to solve the joint optimization problems. Our results showed that the joint design achieves significant performance gain compared to optimizing either wireless power control or fronthaul rate allocation. Besides, we showed that the throughput performance of the proposed simple uniform scalar quantization is very close to the performance upper (cut-set) bound. This has verified that high throughput performance could be practically achieved with C-RAN using simple fronthaul signal quantization methods.

There are also many interesting topics to be studied in the area of fronthaul-constrained OFDMA-based C-RAN system. For instance, the impact of imperfect fronthaul link with packet loss of quantized data; dynamic SC allocation among mobile users; multiple users coexisting on one SC to further improve the spectral efficiency; distributed quantization among RRHs to exploit the signal correlations; and joint wireless resource and fronthaul rate allocations in the downlink, etc.

\begin{appendix}

\subsection{Uniform Scalar Quantization}\label{appendix1}
In this appendix, we provide the details on the implementation of uniform scalar quantization introduced in Section \ref{sec:Uniform Quantization Model}. First, each RRH normalizes the I-branch and Q-branch symbols at each SC into the interval $[-1,1]$ for quantization by the following scaling process:
\begin{align}\label{eqn:scaling}
\bar{y}_{m,n}^\chi=\frac{y_{m,n}^\chi}{\eta_{m,n}^\chi}, ~~~ \chi \in \{I,Q\}, ~~~ \forall m,n.
\end{align}Since $y_{m,n}^I$'s and $y_{m,n}^Q$'s are both real Gaussian random variables the instantaneous power of which can go to infinity in some instances, the probability of overflow should be controlled by a proper selection of the scaling factors $\eta_{m,n}^I$'s and $\eta_{m,n}^Q$'s. In this paper, we apply the so-called ``three-sigma rule'' \cite{Gray98} to select the scaling factors. Specifically, since the average power of $y_{m,n}^I$ and $y_{m,n}^Q$ are both $(|h_{m,k,n}|^2p_{k,n}+\sigma_{m,n}^2)/2$, we set
\begin{align}\label{eqn:scaling factor}
\eta_{m,n}^I=\eta_{m,n}^Q\triangleq\eta_{m,n}=3\sqrt{\frac{|h_{m,k,n}|^2p_{k,n}+\sigma_{m,n}^2}{2}}, ~~~ \forall m,n.
\end{align}As a result, the probability of overflow for both the I-branch and Q-branch symbols is expressed as
\begin{align}\label{eqn:overflow}
P(|\bar{y}_{m,n}^I|>1)=P(|\bar{y}_{m,n}^Q|>1)=2Q(3)=0.0027, ~~~ \forall m,n.
\end{align}Note that in the case of overflow, the quantized value can be set to be $1$ if the scaled symbol is larger than $1$ or $-1$ if it is smaller than $-1$.

Next, RRH $m$ implements uniform quantization on the normalized symbols $\bar{y}_{m,n}^I$'s and $\bar{y}_{m,n}^Q$'s at each SC in the interval $[-1,1]$. We assume that RRH $m$ uses $D_{m,n}\geq 1 $ bits to quantize the symbol received on SC $n$, resulting $2^{D_{m,n}}$ quantization levels, for which the quantization step size is given by\begin{align}\label{eqn:quantization step size}
\Delta_{m,n}=\frac{2}{2^{D_{m,n}}}=2^{1-D_{m,n}}, ~~~ \forall m,n.
\end{align}Furthermore, for each normalized symbol $\bar{y}_{m,n}^I$ or $\bar{y}_{m,n}^Q$, its quantized value is given by
\begin{align}\label{eqn:quantized value}
\check{y}_{m,n}^\chi=\frac{\lceil 2^{D_{m,n}-1}\bar{y}_{m,n}^\chi\rceil}{2^{D_{m,n}-1}}-\frac{1}{2^{D_{m,n}}}, ~~~ \chi\in\{I,Q\}, ~ \forall m,n,
\end{align}where $\lceil x \rceil$ denotes the minimum integer that is no smaller than $x$. Then, $\check{y}_{m,n}^I$'s and $\check{y}_{m,n}^Q$'s are encoded into digital codewords $\hat{y}_{m,n}^I$'s and $\hat{y}_{m,n}^Q$'s and transmitted to the BBU.

\subsection{Proof of Proposition \ref{fronthaul}}\label{appendix2}
Note that the I, Q symbols, i.e., $y_{m,n}^I$'s and $y_{m,n}^Q$'s, are obtained by sampling of the I, Q waveforms, the bandwidth of which is $B/2N$, $\forall m,n$. As a result, at each RRH, the Nyquist sampling rate for the I, Q waveforms at each SC is $B/N$ samples per second. Furthermore, since at RRH $m$, each sample at SC $n$ is represented by $D_{m,n}$ bits, the corresponding transmission rate in the fronthaul link is expressed as
\begin{align}
T_{m,n}^{(U)}=\frac{BD_{m,n}}{N}+\frac{BD_{m,n}}{N}=\frac{2BD_{m,n}}{N}.
\end{align}Then, the overall transmission rate from RRH $m$ to the BBU in the fronthaul link is given as
$T_m^{(U)}=\sum_{n=1}^NT_{m,n}^{(U)}$, which should not exceed the fronthaul link capacity $\bar{T}_m$, $\forall m$. Proposition \ref{fronthaul} is thus proved.

\subsection{Proof of Proposition \ref{rate}}\label{appendix3}
To derive the end-to-end sum-rate, we need to calculate the power of the quantization error given in (\ref{eqn:quantized signal}), i.e., $q_{m,n}$, $\forall m,n$. Note that in (\ref{eqn:quantized signal}) we have $\tilde{y}_{m,n}=\eta_{m,n}(\check{y}_{m,n}^I+j\check{y}_{m,n}^Q)$, $\forall m,n$. According to Widrow Theorem \cite{Franklin90}, if the number of quantization levels (i.e., $2^{D_{m,n}}$) is large, and the signal varies by at least some quantization levels from sample to sample, the quantization noise can be assumed to be uniformly distributed. As a result, we assume that the quantization errors for both I, Q signals, which are denoted by $e_{m,n}^I$ and $e_{m,n}^Q$ with $e_{m,n}=e_{m,n}^I+je_{m,n}^Q$, are uniformly distributed 
in $[-\eta_{m,n}\Delta_{m,n}/2,\eta_{m,n}\Delta_{m,n}/2]$, $\forall m,n$. Then we have
\begin{align}\label{eqn:quantization noise power}
q_{m,n}& =\int_{-\frac{\eta_{m,n} \Delta_{m,n}}{2}}^{\frac{\eta_{m,n} \Delta_{m,n}}{2}}\frac{(e_{m,n}^I)^2}{\eta_{m,n} \Delta_{m,n}}d e_{m,n}^I+\int_{-\frac{\eta_{m,n} \Delta_{m,n}}{2}}^{\frac{\eta_{m,n} \Delta_{m,n}}{2}}\frac{(e_{m,n}^Q)^2}{\eta_{m,n} \Delta_{m,n}}d e_{m,n}^Q \nonumber \\ & =\frac{\eta_{m,n}^2\Delta_{m,n}^2}{6}=3(|h_{m,k,n}|^2p_{k,n}+\sigma_{m,n}^2)2^{-2D_{m,n}} \nonumber \\ & \overset{(a)}{=} 3(|h_{m,k,n}|^2p_{k,n}+\sigma_{m,n}^2)2^{-\frac{NT_{m,n}^{(U)}}{B}},
\end{align}where $(a)$ is obtained by substituting $D_{m,n}$ by $T_{m,n}^{(U)}$ according to (\ref{eqn:fronthaul link uniform quantization SC}). Then according to (\ref{eqn:optimal SINR}), a lower bound for the achievable rate of user $k$ at SC $n$, by viewing $\mv{w}_n^H\mv{e}_n$ given in (\ref{eqn:beamforming}) as the worst-case Gaussian noise (it is worth noting that the equivalent quantization error given in (\ref{eqn:beamforming}), i.e., $\mv{w}_n^H\mv{e}_n$, is the summation of $N$ independent uniform distributed random variables $e_{m,n}$'s. According to the central limit theory, $\mv{w}_n^H\mv{e}_n$ tends to be Gaussian distributed when $N$ is large), can be expressed as
\begin{align}
R_{k,n}^{(U)}&=\frac{B}{N}\log_2\left(1+\sum\limits_{m=1}^M\frac{|h_{m,k,n}|^2p_{k,n}}{\sigma_{m,n}^2+q_{m,n}}\right)\nonumber \\ & =\frac{B}{N}\log_2\left(1+\sum\limits_{m=1}^M\frac{|h_{m,k,n}|^2p_{k,n}}{\sigma_{m,n}^2+3(|h_{m,k,n}|^2p_{k,n}+\sigma_{m,n}^2)2^{-\frac{NT_{m,n}^{(U)}}{B}}}\right).
\end{align}The end-to-end throughput of all users is thus expressed as
$R_{{\rm sum}}^{(U)}=\sum_{k=1}^K \sum_{n\in \Omega_k} R_{k,n}^{(U)}$. Proposition \ref{rate} is thus proved.

\subsection{Proof of Proposition \ref{proposition1}}\label{appendix4}
The Lagrangian of problem (\ref{eqn:p2}) is expressed as
\begin{align}\label{eqn:lagrangian 1}
\mathcal{L}(\{p_n\},\lambda)=\frac{1}{N}\sum\limits_{n=1}^N\log_2\left(1+\frac{|h_n|^2p_n}{\sigma_n^2+\frac{|h_n|^2p_n+\sigma_n^2}{2^{N\hat{T}_n^{(G)}/B}-1}}\right)-\lambda\left(\sum\limits_{n=1}^Np_n-\bar{P}\right),
\end{align}where $\lambda$ is the dual variable associated with the transmit power constraint in problem (\ref{eqn:p2}). Then, the Lagrangian dual function of problem (\ref{eqn:p2}) is expressed as
\begin{align}\label{eqn:dual function 1}
g(\lambda)=\max\limits_{p_n\geq 0, \forall n} \mathcal{L}(\{p_n\},\lambda).
\end{align}The maximization problem (\ref{eqn:dual function 1}) can be decoupled into parallel subproblems all having the same structure and each for one SC. For one particular SC, the associated subproblem is expressed as
\begin{align}\label{eqn:subproblem 1}
\max\limits_{p_n\geq 0} ~ \mathcal{L}_n(p_n),
\end{align}where
\begin{align}\label{eqn:power 1}
\mathcal{L}_n(p_n)=\frac{1}{N}\log_2\left(1+\frac{|h_n|^2p_n}{\sigma_n^2+\frac{|h_n|^2p_n+\sigma_n^2}{2^{N\hat{T}_n^{(G)}/B}-1}}\right)-\lambda p_n, ~~~ n=1,\cdots,N.
\end{align}It can be shown that $\mathcal{L}_n(p_n)$ is concave over $p_n$, $\forall n$. The derivative of $\mathcal{L}_n(p_n)$ over $p_n$ is expressed as
\begin{align}\label{eqn:derivative 1}
\frac{\partial \mathcal{L}_n(p_n)}{\partial p_n}=\frac{|h_n|^2\sigma_n^2\left(2^{\frac{N\Hat{T}_n^{(G)}}{B}}-1\right)}{\left(|h_n|^2p_n+\sigma_n^22^{\frac{N\Hat{T}_n^{(G)}}{B}}\right)(|h_n|^2p_n+\sigma_n^2)N\ln 2}-\lambda, ~~~ \forall n.
\end{align}By setting $\partial \mathcal{L}_n(p_n)/\partial p_n=0$, we have
\begin{align}\label{eqn:equation}
p_n^2+\alpha_n p_n +\eta_n=0, ~~~ n=1,\cdots,N,
\end{align}where $\alpha_n$'s and $\eta_n$'s are given in (\ref{eqn:alpha}) and (\ref{eqn:eta}), respectively. If $\eta_n<0$, then there exists a unique positive solution to the quadratic equation (\ref{eqn:equation}), denoted by $\tilde{p}_n=(-\alpha_n+\sqrt{\alpha_n^2-4\eta_n})/2$. In this case, $\mathcal{L}_n(p_n)$ is an increasing function over $p_n$ in the interval $(0,\tilde{p}_n)$, and decreasing function in the interval $[\tilde{p}_n,\infty)$. As a result, $\mathcal{L}_n(p_n)$ is maximized when $p_n=\tilde{p}_n$. Otherwise, if $\eta\geq 0$, there is no positive solution to the quadratic equation (\ref{eqn:equation}), and thus $\mathcal{L}_n(p_n)$ is a decreasing function over $p_n$ in the interval $(0,\infty)$. In this case, $\mathcal{L}_n(p_n)$ is maximized when $p_n=0$.

After problem (\ref{eqn:dual function 1}) is solved given any $\lambda$, in the following we explain how to find the optimal dual solution for $\lambda$. It can be shown that the objective function in problem (\ref{eqn:p2}) is an increasing function over $\{p_n\}$, and thus the transmit power constraint must be tight in problem (\ref{eqn:p2}). As a result, the optimal $\lambda$ can be efficiently obtained by a simple bisection method such that the transmit power constraint is tight in problem (\ref{eqn:p2}). Proposition \ref{proposition1} is thus proved.

\subsection{Proof of Proposition \ref{proposition2}}\label{appendix5}

Let $\beta$ denote the dual variable associated with the fronthaul link capacity constraint in problem (\ref{eqn:p3}). Similar to Appendix \ref{appendix1}, it can be shown that problem (\ref{eqn:p3}) can be decoupled into the $N$ subproblems with each one formulated as
\begin{align}\label{eqn:subproblem 2}
\max\limits_{T_n^{(G)}\geq 0} ~ \mathcal{L}_n(T_n^{(G)}),
\end{align}where
\begin{align}\label{eqn:fronthaul rate 2}
\mathcal{L}_n(T_n^{(G)})=\frac{1}{N}\log_2\left(1+\frac{|h_n|^2\hat{p}_n}{\sigma_n^2+\frac{|h_n|^2\hat{p}_n+\sigma_n^2}{2^{NT_n^{(G)}/B}-1}}\right)-\beta T_n, ~~~ n=1,\cdots,N.
\end{align}The derivative of $\mathcal{L}_n(T_n^{(G)})$ over $T_n^{(G)}$ is expressed as
\begin{align}
\frac{\partial \mathcal{L}_n(T_n^{(G)})}{\partial T_n^{(G)}}=\frac{1}{B}-\beta-\frac{\sigma_n^2 2^{\frac{NT_n^{(G)}}{B}}}{B\left(|h_n|^2\hat{p}_n+\sigma_n^2 2^{\frac{NT_n^{(G)}}{B}}\right)}, ~~~ n=1,\cdots,N.
\end{align}If $\beta\geq \frac{1}{B}$, then $\frac{\partial \mathcal{L}_n(T_n^{(G)})}{\partial T_n^{(G)}}\leq 0$, i.e., $\mathcal{L}_n(T_n^{(G)})$ is a decreasing function over $T_n^{(G)}$, $\forall n$. In this case, we have $\bar{T}_n^{(G)}=0$, $\forall n$, which cannot be the optimal solution to problem (\ref{eqn:p3}). As a result, the optimal dual solution must satisfy $\beta<\frac{1}{B}$. In this case, it can be shown that $\mathcal{L}_n(T_n^{(G)})$ is an increasing function over $T_n^{(G)}$ when $T_n^{(G)}<\frac{B}{N}\log_2\frac{(1-\beta B)|h_n|^2\hat{p}_n}{\beta B \sigma_n^2}$, and decreasing function otherwise. As a result, $\mathcal{L}_n(T_n^{(G)})$ is maximized at $T_n^{(G)}=\max(\frac{B}{N}\log_2\frac{(1-\beta B)|h_n|^2\hat{p}_n}{\beta B \sigma_n^2},0)$.
After problem (\ref{eqn:subproblem 2}) is solved given any $\beta<\frac{1}{B}$, the optimal $\beta$ that is the dual solution to problem (\ref{eqn:p3}) can be efficiently obtained by a simple bisection method over $(0,\frac{1}{B})$ such that the fronthaul link capacity constraint is tight in problem (\ref{eqn:p3}). Proposition \ref{proposition2} is thus proved.

\subsection{Proof of Proposition \ref{proposition3}}\label{appendix6}
With constraints given in (\ref{eqn:constraint1}), $R_{{\rm sum}}^{(G)}$ given in (\ref{eqn:test channel rate}) reduces to
\begin{align}
\frac{R_{{\rm sum}}^{(G)}}{B}=\frac{1}{N}\sum\limits_{n=1}^N\log_2\left(1+\frac{|h_n|^2p_n}{2\sigma_n^2}\right).
\end{align}Moreover, it can be shown from (\ref{eqn:constraint1}) that
\begin{align}
\frac{T^{(G)}}{B}=\frac{1}{N}\sum\limits_{n=1}^N\log_2\left(2+\frac{|h_n|^2p_n}{\sigma_n^2}\right)=\frac{R_{{\rm sum}}^{(G)}}{B}+1.
\end{align}Thereby, with the additional constraints given in (\ref{eqn:constraint1}), problem (P1) can be simplified as the following power control problem.
\begin{align}\mathop{\mathtt{Maximize}}_{\{p_n\}} & ~~~ \frac{1}{N}\sum\limits_{n=1}^N\log_2\left(1+\frac{|h_n|^2p_n}{2\sigma_n^2}\right) \nonumber \\
\mathtt {Subject \ to} & ~~~ \frac{1}{N}\sum\limits_{n=1}^N\log_2\left(1+\frac{|h_n|^2p_n}{2\sigma_n^2}\right)+1\leq \frac{\bar{T}}{B} \nonumber \\ & ~~~ \sum\limits_{n=1}^Np_n\leq \bar{P}. \label{eqn:fix quantization in P1}
\end{align}Let $\{\check{p}_n\}$ and $\{\tilde{p}_n\}$ denote the optimal power solution to problem (\ref{eqn:fix quantization in P1}) and the relaxed version of problem (\ref{eqn:fix quantization in P1}) without the first fronthaul link constraint, respectively. If $(1/N)\sum_{n=1}^N\log_2(1+|h_n|^2\tilde{p}_n/2\sigma_n^2)+1\leq \bar{T}/B$, we have $\check{p}_n=\tilde{p}_n$, $\forall n$. Otherwise, it can be shown that any feasible solution to the following problem is optimal to problem (\ref{eqn:fix quantization in P1}):
\begin{align}\mathop{\mathtt{Find}} & ~~~ \{p_n\} \nonumber \\
\mathtt {Subject \ to} & ~~~ \frac{1}{N}\sum\limits_{n=1}^N\log_2\left(1+\frac{|h_n|^2p_n}{2\sigma_n^2}\right)+1= \frac{\bar{T}}{B} \nonumber \\ & ~~~ \sum\limits_{n=1}^Np_n\leq \bar{P}. \label{eqn:fix quantization in P3}
\end{align}To summarize, the cut-set bound based optimal value of problem (\ref{eqn:fix quantization in P1}) is expressed as
\begin{align}\label{eqn:optimal value}
\frac{\bar{R}_{{\rm sum}}^{(G)}}{B}=\min\left\{\frac{1}{N}\sum\limits_{n=1}^N\log_2\left(1+\frac{|h_n|^2\tilde{p}_n}{2\sigma_n^2}\right),\frac{\bar{T}}{B}-1\right\}.
\end{align}In the following, we compare this optimal value with the capacity upper bound $C$ given in (\ref{eqn:capacity upper bound}). First, we have
\begin{align}
\frac{1}{N}\sum\limits_{n=1}^N\log_2\left(1+\frac{|h_n|^2p_n^{{\rm wf}}}{\sigma_n^2}\right)-1 & <
\frac{1}{N}\sum\limits_{n=1}^N\log_2\left(1+\frac{|h_n|^2p_n^{{\rm wf}}}{2\sigma_n^2}\right) \nonumber \\ &\overset{(a)}{\leq}
\frac{1}{N}\sum\limits_{n=1}^N\log_2\left(1+\frac{|h_n|^2\tilde{p}_n}{2\sigma_n^2}\right),
\end{align}where $(a)$ is because $\{\tilde{p}_n\}$ is the optimal power solution to problem (\ref{eqn:fix quantization in P1}) without the fronthaul link constraint. It then follows that
\begin{align}
\frac{\bar{R}_{{\rm sum}}^{(G)}}{B}&=\min\left\{\frac{1}{N}\sum\limits_{n=1}^N\log_2\left(1+\frac{|h_n|^2\tilde{p}_n}{2\sigma_n^2}\right),\frac{\bar{T}}{B}-1\right\} \nonumber \\ &\geq \min\left\{\frac{1}{N}\sum\limits_{n=1}^N\log_2\left(1+\frac{|h_n|^2p_n^{{\rm wf}}}{\sigma_n^2}\right)-1,\frac{\bar{T}}{B}-1\right\} = C-1.
\end{align}Proposition \ref{proposition3} is thus proved.

\subsection{Proof of Proposition \ref{proposition4}}\label{appendix7}
First, it follows that
\begin{align}
T^{(U)}=\frac{B}{N}\sum\limits_{n=1}^N\log_2\left(1+\frac{|h_n|^2\check{p}_n}{\sigma_n^2}\right) < \frac{B}{N}\sum\limits_{n=1}^N\log_2\left(2+\frac{|h_n|^2\check{p}_n}{\sigma_n^2}\right) \leq \bar{T}.
\end{align}As a result, $\{p_n=\check{p}_n,T_n^{(G)}=(B/N)\log_2(1+|h_n|^2\check{p}_n/\sigma_n^2)\}$ is a feasible solution to problem (P2-NoInt). Furthermore, with $\{p_n=\check{p}_n,T_n^{(G)}=(B/N)\log_2(1+|h_n|^2\check{p}_n/\sigma_n^2)\}$, $R_{{\rm sum}}^{(U)}$ given in (\ref{eqn:uniform quantization sum-rate}) reduces to
\begin{align}
R_{{\rm sum}}^{(U)}=\frac{B}{N}\sum\limits_{n=1}^N\log_2\left(1+\frac{|h_n|^2\check{p}_n}{4\sigma_n^2}\right).
\end{align}It then follows that
\begin{align}
\frac{\bar{R}_{{\rm sum}}^{(U)}}{B} > \frac{1}{N}\sum\limits_{n=1}^N\log_2\left(1+\frac{|h_n|^2\check{p}_n}{2\sigma_n^2}\right)-1 = \frac{\bar{R}_{{\rm sum}}^{(G)}}{B}-1.
\end{align}Proposition \ref{proposition4} is thus proved.

\subsection{Proof of Lemma \ref{lemma1}}\label{appendix8}
Define
\begin{align}
\varphi_{m,k,n}(p_{k,n})=\frac{|h_{m,k,n}|^2p_{k,n}\hat{\psi}_{m,n}}{\sigma_{m,n}^2\hat{\psi}_{m,n}+|h_{m,k,n}|^2p_{k,n}+\sigma_{m,n}^2}, ~~~ {\rm if} ~ n\in \Omega_k, ~ \forall m,n.
\end{align}Then, it can be shown that $\varphi_{m,k,n}(p_{k,n})$ is concave over $p_{k,n}$, $\forall m,n$. As a result, $\sum_{m=1}^M\varphi_{m,k,n}(p_{k,n})$ is concave over $p_{k,n}$, $\forall k,n$. According to the composition rule \cite{Boyd04}, $\log_2(1+\sum_{m=1}^M\varphi_{m,k,n}(p_{k,n}))$ is concave over $p_{k,n}$, $\forall k,n$. It then follows that the objective function of problem (\ref{eqn:p5}), i.e., $\sum_{k=1}^K\sum_{n\in \Omega_k} \log_2(1+\sum_{m=1}^M\varphi_{m,k,n}(p_{k,n}))$, is concave over $\{p_{k,n}\}$. Lemma \ref{lemma1} is thus proved.

\subsection{Proof of Lemma \ref{lemma2}}\label{appendix9}
First, due to the inequality given in (\ref{eqn:first order approximation}), any feasible solution to problem (\ref{eqn:p7}) must be a feasible solution to problem (\ref{eqn:p6}). Thereby, $\check{\psi}_{m,n}$'s must be feasible to problem (\ref{eqn:p6}). Next, it can be observed that if $\tilde{\psi}_{m,n}$'s is feasible to problem (\ref{eqn:p6}), it must be feasible to problem (\ref{eqn:p7}). Since $\check{\psi}_{m,n}$'s is the optimal solution to problem (\ref{eqn:p7}), the sum-rate achieved by it must be no smaller than that achieved by $\tilde{\psi}_{m,n}$'s. Lemma \ref{lemma2} is thus proved.

\end{appendix}


\linespread{1.1}

\end{document}